\documentclass[
notitlepage,
floatfix,
aps,pra,
reprint,
twocolumn,
superscriptaddress
]{revtex4-2}
\usepackage[utf8]{inputenc}
\usepackage[normalem]{ulem}
\usepackage[dvipsnames]{xcolor}
\usepackage{times}
 \usepackage{braket}
\usepackage{listings}
\usepackage{wrapfig}
 
\usepackage{makecell}
\usepackage{tabularx, tabularray}
\usepackage{longtable}
\usepackage{multirow}
\usepackage[lmargin=.7in,rmargin=.7in,tmargin=.7in,bmargin=1in]{geometry}

\usepackage{fancyhdr}
\newcommand{\suppl}{Supplemental Materials}
\usepackage{makecell}

\DeclareMathAlphabet{\mathmybb}{U}{bbold}{m}{n}
\usepackage{mathtools}
\usepackage{titlesec}
\usepackage{ulem}

\usepackage{graphicx}
\usepackage{tikz}
\usepackage{float}
\usepackage{xcolor}
\usepackage{hyperref}
\hypersetup{colorlinks,linkcolor={blue},citecolor={blue},urlcolor={blue}}

\usepackage{enumerate}

\renewcommand\ket[1]{|#1\rangle}
\renewcommand{\bra}[1]{\langle{#1}|}
\renewcommand{\braket}[2]{\langle{#1}|{#2}\rangle}
\newcommand{\ketbra}[2]{|{#1}\rangle\!\langle{#2}|}
\newcommand{\dm}[1]{\ketbra{#1}{#1}}

\newcommand{\HSnorm}[1]{\left\lVert #1 \right\rVert_{\rm HS}}
\newcommand{\Tr}{\mathrm{Tr}}
\newcommand{\1}{\mathmybb{1}}
\newcommand{\h}{\hat{H}}
\newcommand{\W}{\mathcal{\hat{W}}}
\def\E{\overline{E}}
\def\V{\overline{V}}
\def\W{\hat{W}}
\newcommand{\abso}[1]{\left| #1 \right|}

\usepackage{amsmath}  

\usepackage{appendix}

\usepackage{amsthm}   
\usepackage{amssymb}  
\newtheorem{theorem}{Theorem}
\newtheorem{lemma}[theorem]{Lemma}
\newtheorem{corollary}[theorem]{Corollary}
\newtheorem{proposition}[theorem]{Proposition}

\newtheorem*{remark}{Remark}
\makeatletter
\def\maketitle{
	\@author@finish
	\title@column\titleblock@produce
	\suppressfloats[t]}
\makeatother

\newcommand{\mg}[1]{{\color{orange}}}

\definecolor{bhblue}{HTML}{0096FF}

\definecolor{tableblue}{HTML}{ccffff}

\begin{document}
\title{Double-bracket quantum algorithms for quantum imaginary-time evolution }

\author{Marek Gluza}
\email{marekludwik.gluza@ntu.edu.sg}
\affiliation{School of Physical and Mathematical Sciences, Nanyang Technological University, 637371, Singapore}
\author{Jeongrak Son}
\affiliation{School of Physical and Mathematical Sciences, Nanyang Technological University, 637371, Singapore}
\author{Bi Hong Tiang}
\affiliation{School of Physical and Mathematical Sciences, Nanyang Technological University, 637371, Singapore}
\newcommand\FOKUS{Fraunhofer Institute FOKUS, Berlin, Germany}
\author{René Zander} 
\affiliation{\FOKUS}
\author{Raphael Seidel} 
\affiliation{\FOKUS}

\newcommand{\EPFL}{Institute of Physics, Ecole Polytechnique Fédérale de
Lausanne (EPFL), Lausanne, Switzerland}
\newcommand{\KEIO}{Quantum Computing Center, Keio University, Hiyoshi 3-14-1, Kohoku-ku, Yokohama 223-8522, Japan}
\author{Yudai Suzuki}
\affiliation{\EPFL}\affiliation{\KEIO}
\author{Zo\"{e} Holmes}
\affiliation{\EPFL}

\author{Nelly H. Y. Ng }
\email{nelly.ng@ntu.edu.sg}
\affiliation{School of Physical and Mathematical Sciences, Nanyang Technological University, 637371, Singapore}
\affiliation{Centre for Quantum Technologies, National University of Singapore, 3 Science Drive 2, 117543, Singapore}
\date{\today}

\begin{abstract}
Efficiently preparing approximate ground-states of large, strongly correlated systems on quantum hardware is challenging and yet nature is innately adept at this. This has motivated the study of thermodynamically inspired approaches to ground-state preparation that aim to replicate cooling processes via imaginary-time evolution. However, synthesizing quantum circuits that efficiently implement imaginary-time evolution is itself difficult, with prior proposals generally adopting heuristic variational approaches or using deep block encodings. Here, we use the insight that quantum imaginary-time evolution is a solution of Brockett's double-bracket flow and synthesize circuits that implement double-bracket flows coherently on the quantum computer. 
We prove that our Double-Bracket Quantum Imaginary-Time Evolution (DB-QITE) algorithm inherits the cooling guarantees of imaginary-time evolution. Concretely, each step is guaranteed to i) decrease the energy of an initial approximate ground-state by an amount proportion to the energy fluctuations of the initial state and ii) increase the fidelity with the ground-state. 
We provide gate counts for DB-QITE through numerical simulations in Qrisp which demonstrate scenarios where DB-QITE outperforms quantum phase estimation.
Thus DB-QITE provides a means to systematically improve the approximation of a ground-state using shallow circuits. 
\end{abstract}

\maketitle

\textit{Introduction.---}Preparing ground-states of Hamiltonians is a fundamental task in quantum computation with wide-ranging applications from studying properties of materials~\cite{quantumcentricmaterialscience} and chemicals~\cite{RMP_QC_chemistry} 
to solving optimization problems~\cite{blekos2024review}.
However, ground-state preparation is not only NP-hard~\cite{barahona1982computational} but also QMA-complete~\cite{GottesmanQMA,QMA,Bausch,Kempe-QIC-2003}, and thus is a challenging problem even for quantum computers~\cite{osborne2012hamiltonian,gharibian2015quantum,QuantumHamiltonianComplexityAharonov,Aharonov-2002}, let alone classical ones. Nevertheless, improvements on the computational efficiency over classical simulation is anticipated to be feasible with quantum processors~\cite{PhysRevB.phasecraft_vqe,wu2024variational,AlexiaEnergyInitiative}. In other words, the shift to quantum computing could offer a pathway to overcoming practical~\cite{NatafMilaSUN}, conceptual~\cite{Feynman} and fundamental~\cite{ETH_SAT,strassen1969gaussian} challenges to the classical simulation of ground-states. 

To date, various ground-state preparation algorithms have been proposed for  both  fault-tolerant \cite{nearoptimalground_Lin_2020,ge2019faster,gilyen2019quantum,fragmented_QITE2024,dong2022ground,epperly2021theory,temme2011quantum,kitaev1995quantum,brassard2002quantum,tan2020quantum,Low2019hamiltonian,PhysRevA.90.022305,double_bracket2024,PoulinWocjan,motlagh2024ground,Kirby2023exactefficient} and near-term quantum computers~\cite{avqite,mcArdle2019,nishi2021,yeter2020practical,motta2020determining,yeter2021benchmarking,Gomes2020,Huang2023,PhysRevResearch.6.033084,LinQPE,larose2019variational,Zeng_2021,parrish2019quantum,stair2020multireference,kokail2019self,google2020hartree,riemannianflowPhysRevA.107.062421,robbiati2024double,xiaoyue2024strategies,KishorRevModPhys.94.015004,cerezo2021variational,PreparationMPS2,PreparationMPS}.
Among these, a promising thermodynamic-inspired approach to cool an initial state $\ket{\Psi_0}$ with respect to a Hamiltonian $\h$ is the \textit{imaginary-time evolution} (ITE) defined by
\begin{align}
	|\Psi(\tau) \rangle = \frac{e^{-\tau \hat H} |\Psi_0\rangle}  {\| {e^{-\tau \hat H}}|\Psi_0\rangle \| }  \ .
 \label{eq:QITE}
\end{align}
Here $\tau$ is the ITE duration and the normalization involves the norm defined for any vector $\ket \Omega$ by $\|\ket{\Omega}\|=\sqrt{\langle\Omega|\Omega\rangle}$.  ITE is guaranteed to converge to the ground-state $\ket{\lambda_0}$ of $\h$ in case the initial overlap with the ground-state is non-zero~\cite{GellmannLow}.

We distinguish ITE from \emph{quantum imaginary-time evolution (QITE)}~\cite{motta2020determining,avqite,mcArdle2019,nishi2021,yeter2020practical,yeter2021benchmarking,Gomes2020,Huang2023,fragmented_QITE2024} in that ITE is defined by the normalized action of a non-unitary propagator and QITE is the implementation of ITE by explicitly using a unitary $Q_\tau$ such that $|\Psi(\tau) \rangle = Q_\tau |\Psi_0 \rangle$.
Finding the unitaries that implement the ITE states $\ket{\Psi(\tau)}$ is not straightforward.
One family of approaches use a hybrid quantum-classical optimization loop to learn $Q_\tau$ ~\cite{motta2020determining,avqite,mcArdle2019,nishi2021,yeter2020practical,yeter2021benchmarking,Gomes2020}. This can yield compressed circuits by fine-tuning to individual input instances, but scaling to large problem sizes is generally inhibited by growing requirements on measurement precision~\cite{larocca2024review, cerezo2023does}.
Another approach is to extend the system size and approximate the non-unitary propagator with qubitization~\cite{fragmented_QITE2024}. However, the overheads of implementing so-called block-encodings preclude flexible near-term experiments.
In other words, constructing efficient circuits for QITE remains an open problem.

\begin{figure}[t]
\centering
\begin{tikzpicture}
\definecolor{lightgray}{HTML}{F4F4F4}
\definecolor{littlelightgray}{HTML}{ecececff}
\definecolor{pale}{HTML}{7ca3d4ff}
\definecolor{lightred}{HTML}{d8a2a2}
\node[rotate=90,text width=0.3cm] at (-4.0,0.8) 
    {QITE};
\node[anchor=center] (russell) at (-0.2,1)
    {\centering\includegraphics[width=0.35\textwidth]{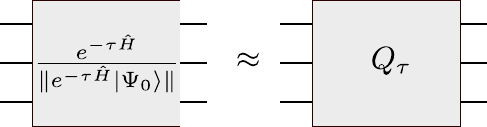}};
\node[rotate=90,text width=1.9cm] at (-4.0,-0.8) 
    {DB-QITE};
    \node[anchor=center] (russell) at (0.,-0.9)
    {\centering\includegraphics[width=0.4\textwidth]{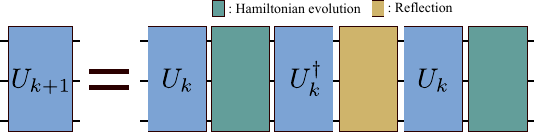}};

\end{tikzpicture}
\caption{\textbf{Double-bracket Quantum Imaginary-Time Evolution (DB-QITE).}
We propose a new quantum algorithm to implement imaginary-time evolution (ITE). To implement the Quantum Imaginary-Time Evolution (QITE) unitary $Q_{\tau}$, we utilize a Double-Bracket Quantum Algorithm (DBQA) and show that QITE can be recursively compiled using Hamiltonian evolution and reflection gates. }
\label{fig:summary}
\end{figure}

In this work, we offer a resolution to this problem by drawing on the observation that ITE is a solution to well-studied differential equations known as \emph{double-bracket flows (DBF)}~\cite{bloch1990steepest,moore1994numerical,Brockett1991DBF,smith1993geometric,optimization2012,BLOCH1985103,bloch1992completely,BROCKETT1989761,bloch1985completely,deift1983ordinary,Chu_iterations,wegner1994flow,wegner2006flow,hastings2022lieb,GlazekWilson,GlazekWilson2,kehrein_flow,brockett2005smooth}. 
DBFs are appealing for quantum computation because their solutions arise from unitaries, a feature recently exploited for quantum circuit synthesis by means of double-bracket quantum algorithms (DBQA)~\cite{double_bracket2024,robbiati2024double,xiaoyue2024strategies}. 
DBFs come with local optimality proofs in that they implement gradient flows~\cite{Brockett1991DBF,bloch1990steepest,moore1994numerical,smith1993geometric,optimization2012,malvetti2024randomized} and these are known to converge exponentially fast to their fixed points~\cite{karimi2016linear}.
Accordingly, DBQAs have recently been shown to provide a compilation procedure which is able to significantly improve on circuits optimized by brute-force~\cite{robbiati2024double}. Here we introduce an algorithm, termed \textit{double-bracket QITE} (DB-QITE), that implements steps of gradient
flows \textit{coherently} on a quantum computer~\cite{double_bracket2024}, without resorting to classical computations, heuristic hybrid quantum-classical variational methods, or block encodings. 
Instead, the DBQA approach produces recursive quantum circuits which at every step approximates the DBF of ITE, as sketched in Fig.~\ref{fig:summary}.

Crucially, DB-QITE inherits the cooling properties of ITE, as summarized in Table~\ref{tab:summary}. Namely, we provide rigorous guarantees that DB-QITE systematically lowers the energy of a state and increases its fidelity with the ground-state.
As we will see, DB-QITE has a similar cooling rate as ITE, with the rate larger for quantum states with high energy fluctuations. Moreover, a single step of DB-QITE is guaranteed to increase the fidelity with the ground-state by an amount proportional to the initial fidelity with the ground-state and the spectral gap of the target Hamiltonian. Furthermore, as our bounds hold for any input state (in contrast to Ref.~\cite{AnshuImproved}) this argument can be applied iteratively to show that the DB-QITE unitaries yield states which converge exponentially fast in the number of iterations to the ground-state. 
Given this recursive structure, the circuit depths required to implement DB-QITE grow exponentially with the number of steps.
As we show numerically, optimizing the free parameter of step durations reduces the number of iterations to a level competitive with state-of-the-art methods.

Thus DB-QITE provides a means to systematically reduce the energy of an approximate ground-state.
We expect our algorithm to be used both as a standalone ground-state preparation method in the early fault-tolerant era, as well as in conjunction with more established~\cite{nearoptimalground_Lin_2020,fragmented_QITE2024,dong2022ground,gilyen2019quantum,ge2019faster} and heuristic approaches~\cite{avqite,mcArdle2019,nishi2021,motta2020determining,yeter2020practical,MottaAdiabatic,KishorRevModPhys.94.015004,cerezo2021variational} to ground-state preparation.

\medskip

\begin{table}[t]
\centering
\begin{tikzpicture}
\definecolor{lightgray}{HTML}{F4F4F4}
\definecolor{littlelightgray}{HTML}{ecececff}
\definecolor{pale}{HTML}{7ca3d4ff}
\definecolor{lightred}{HTML}{d8a2a2}
\node[text width=6cm] at (-1.5,-4.7) 
    { };
\draw[thick] (-4.0, -4.6) rectangle (3.7, -7.6);
\fill[littlelightgray] (-4.0, -5.4) rectangle (3.7, -6.5);
\fill[pale] (-4.0, -6.5) rectangle (3.7, -7.6);
\draw[thick] (-4.0, -5.4) -- (3.7, -5.4);
\draw[thick] (-4.0, -6.5) -- (3.7, -6.5);
\draw[thick] (-2.8, -4.6) -- (-2.8, -7.6);
\draw[thick] (1.3, -4.6) -- (1.3, -7.6);
\node[text width=3] at (-3.65,-6.0) 
    {ITE};
\node[align=center,text width=3] at (-3.75,-7.1) 
    {DB-QITE};
\node[align=center,text width=4.0cm] at (-0.7,-4.98) 
    {Fluctuation-Refrigeration Relation};
\node[align=center,text width=4cm] at (2.5,-4.98) 
    {Fidelity};
\node[align=center,text width=4cm] at (-0.7,-6.0) 
    {$\partial_\tau E(\tau) = -2 V(\tau)$\\ $[$Eq.~\eqref{eq:fluctuationRefrigeration}$]$ };
\node[align=center,text width=4cm] at (-0.75,-7.1) 
    {$E_{k+1} \le  E_k - 2s_k V_k + \mathcal O (s_k^{2})$\\ $[$Theorem~\ref{th: fluctuation-refrigeration main}$]$ };
\node[align=center,text width=4cm] at (2.5,-6.0) 
    {$1-\mathcal{O}(e^{-2\tau\Delta})$\\ $[$Eq.~\eqref{eq: fidelity for Hasting QITE}$]$ };
\node[align=center,text width=4cm] at (2.5,-7.1) 
    {$1-\mathcal{O}(q^{k})$\\ $[$Theorem~\ref{th: fidelity convergence}$]$};
\end{tikzpicture}
\caption{\textbf{Energy reduction guarantees for DB-QITE}.
It is known that the ITE can lower the energy $E(\tau)$ by an amount proportional to the energy fluctuation $V(\tau)$. We call this a fluctuation-refrigeration relation.
ITE is also guaranteed to improve the fidelity to the ground-state exponentially in the imaginary-time $\tau$. 
Similarly to ITE, DB-QITE agrees with the amount of energy reduction $E_{k+1}-E_{k}$ up to corrections of order $\mathcal O(s_{k}^2)$, and improves the fidelity exponentially fast in the time of algorithm iterations $k$. }
\label{tab:summary}
\end{table}

\textit{Fluctuation-refrigeration relation of ITE.---}ITE defined in Eq.~\eqref{eq:QITE} satisfies~\cite{GellmannLow}
\begin{align}
    \partial_\tau {|\Psi(\tau) \rangle} 
    &= -(\h  -  E(\tau))|\Psi(\tau)\rangle\ 
\label{eq:GML}
\end{align}
with the energy $E(\tau) = \langle \Psi(\tau) | \h | \Psi(\tau)\rangle$.
Under such evolutions, the state is guaranteed to converge to the ground-state of the Hamiltonian in the limit of $\tau\to\infty$ as long as the initial state $\ket{\Psi_0}$ has some non-zero overlap to the ground-state~\cite{GellmannLow}.
This fact has driven the exploration of how to implement ITE on classical~\cite{paeckel2019time} and quantum computers~\cite{motta2020determining,avqite,mcArdle2019,nishi2021,yeter2020practical,yeter2021benchmarking,Gomes2020}.
From Eq.~\eqref{eq:GML}, a direct computation (see \suppl ~\ref{sec: DBF cooling rate}) gives 
\begin{align}
    \partial_\tau E(\tau) = -2 V(\tau)\ ,    \label{eq:fluctuationRefrigeration}
\end{align}
where $    V(\tau) = \langle \Psi(\tau) | (\h - E(\tau))^2| \Psi(\tau) \rangle
$ is the energy fluctuation (the operator variance of the Hamiltonian).
It follows that higher energy fluctuations in the state lead to a faster energy decrease. We call this a \emph{fluctuation-refrigeration relation}, and will show that an analogous relation holds for the quantum circuits that we will propose to approximate ITE. 

\medskip

\textit{ITE is a solution of DBF.---}We next observe that ITE Eq.~\eqref{eq:GML} can be rewritten as
\begin{align}
    \partial_\tau {|\Psi(\tau) \rangle} 
    &= [\Psi(\tau),\h] |\Psi(\tau) \rangle\ ,
\label{eq:GML DBF}
\end{align}
where  $\Psi(\tau)  = \ket{\Psi(\tau)}\bra{\Psi(\tau)}$ is the density matrix of the ITE state and $[A,B]=AB-BA$ is the commutator~\footnote{We thank D. Gosset for pointing out the appearance of the commutator for QITE which he discovered while working on Ref.~\cite{AnshuImproved}.}.
In terms of the density matrix $\Psi(\tau)$, we get
\begin{align}\label{eq: DBF DE}
     \frac{\partial \Psi(\tau)}{\partial \tau} = \big[ [\Psi(\tau),\h],\Psi(\tau)\big]\ .
\end{align}
This is exactly the form of the well-studied Brockett's DBF~\cite{optimization2012}. 
In End Matter, we review its relation to Riemannian gradient descent on  unitary manifolds and state the relevant cost functions.
Given Eq.~\eqref{eq: DBF DE}, these results from DBF theory now apply to QITE and in essence signify its local optimality. 

\medskip

\textit{QITE from ITE DBF.---}
The difficulty of implementing ITE on quantum computers lies in the non-unitarity of the propagator in Eq.~\eqref{eq:QITE} and its state-dependence in Eq.~\eqref{eq:GML}.
The challenge thus lies in circuit synthesis of unitaries $Q_\tau$ satisfying
\begin{align}
    |\Psi(\tau) \rangle \approx Q_\tau|\Psi_0 \rangle \ .
    \label{QITEsynthesis}
\end{align}
This transition from non-unitary ITE propagation in Eq.~\eqref{eq:QITE} to unitary implementation in Eq.~\eqref{QITEsynthesis} is termed QITE. Intuitively, we have from Eq.~\eqref{eq:GML DBF} that for short durations $\Delta \tau$,
\begin{align}
    \ket{\Psi(\Delta \tau)}\approx e^{\Delta \tau[\Psi(0),\h]}\ket{\Psi_0} \ .
    \label{QITE DBF apprx}
\end{align}
Thus, writing ITE as a DBF directly results in a proposal for QITE, at least for short durations.
It is not straightforward to rigorously quantify the approximation in Eq.~\eqref{QITE DBF apprx}, see Refs.~\cite{double_bracket2024,mcmahon2025equating}.
Nevertheless, we show in Prop.~\ref{prop: DBI frr} of App.~\ref{app: DBI cooling rate} that the energy decrease of this dynamics agrees with Eq.~\eqref{eq:fluctuationRefrigeration} up to corrections of order $\mathcal O(\Delta \tau^2)$. Motivated by this, we initialize a recursion by the initial state $\ket{\sigma_0} = \ket{\Psi_0}$ and 
define 
\begin{align}\label{eq: DBI state main}
    \ket{\sigma_{k+1}} = e^{s_k[\ket{\sigma_k}\bra{\sigma_k},\h]} \ket{\sigma_k}\ .
\end{align}
Here, we denote the time step size in the $(k+1)$-th step as $s_k$. 
Rigorous results from Refs.~\cite{optimization2012,moore1994numerical,bloch1990steepest,smith1993geometric} apply to Eq.~\eqref{eq: DBI state main} and prove convergence to the ground-state as $k\rightarrow\infty$. 
The recursion in Eq.~\eqref{eq: DBI state main} provides an explicit recipe for QITE unitaries in Eq.~\eqref{QITEsynthesis}.
However, further circuit synthesis is needed as $e^{s_k[\ket{\sigma_k}\bra{\sigma_k},\h]}$ is challenging to execute directly on a quantum computer.

\medskip

\textit{DB-QITE circuit synthesis.---}Next, we compile QITE unitaries as quantum circuit using the group commutator iterations first introduced in Ref.~\cite{double_bracket2024}. 
In the context of QITE we define the group commutator of  any density matrix $\Omega$ by
\begin{align}\label{eqGC}
    G_s(\Omega) &= 
    e^{i\sqrt{s}\h}e^{i\sqrt{s}\Omega}    
    e^{-i\sqrt{s}\h}
    e^{-i\sqrt{s}\Omega}\ .  
\end{align}
The \suppl~of Ref.~\cite{double_bracket2024} show that
\begin{equation}
    G_s(\Omega)= e^{s[\Omega,\h]} + \mathcal O(s^{3/2})\ ,
\end{equation}
see Refs.~\cite{dawson2006solovay,commutator_approximation_2022,product_formula2013} for further context. 
In order to formulate the circuit synthesis for DB-QITE we first notice
 that when applying $G_s(\ket \Omega \bra\Omega)$ to $\ket \Omega$ there is a cancellation and find 
\begin{align}\label{eq: GCI initial definition}
    e^{i\sqrt{s}}G_s(\ket \Omega \bra\Omega) \ket{\Omega} = e^{i\sqrt{s}\h}e^{i\sqrt{s}\ket \Omega \bra\Omega}    
    e^{-i\sqrt{s}\h}\ket{\Omega}\ .
\end{align}
This allows us to arrive at the following DB-QITE states
\begin{align}\label{eq: GCI state main}
    \ket{\omega_{k+1}} = e^{i\sqrt{s_k}\h} e^{i\sqrt{s_k}\ket{\omega_k}\bra {\omega_k}}   
    e^{-i\sqrt{s_k}\h}\ket{\omega_k}\ .
\end{align}
Finally, let $U_k$ denote the circuit to prepare $\ket{\omega_k}$ from a trivial reference state $\ket 0$, i.e., $\ket{\omega_k} = U_k |0\rangle$. We can now use unitarity to simplify $e^{i\sqrt{s}\ket{\omega_k}\bra {\omega_k}}= U_k e^{i\sqrt{s_k}\ket 0\langle0|}    U_k^\dagger$ and arrive at a recursive formula for DB-QITE circuit synthesis
    \begin{align}
    U_{k+1} = e^{i\sqrt{s_k}\h}U_k e^{i\sqrt{s_k}\ket 0\langle0|}    U_k^\dagger
  e^{-i\sqrt{s_k}\h} U_k\ \ .
  \label{DB-QITE Uk}
\end{align}
Step $k=0$ involves $U_0$ which  is used to define $\ket{\omega_0}:=U_0\ket 0$. 
When $U_0= I$ then DB-QITE operates as standalone and when $\ket{\omega_0}$ approximates the ground-state $\ket{\lambda_0}$ of the input Hamiltonian $\h$ then $U_0$ is a warm-start with subsequent steps of DB-QITE continuing to leverage on $U_0$ implicitly. 
Fig.~\ref{fig:summary} summarizes this recursive circuit synthesis and the compilation of DB-QITE is completed by invoking known efficient schemes for  Hamiltonian simulation~\cite{childs2010limitations,kothari,childs2012hamiltonian,PhysRevX.TrotterSuzukiError,SuChilds_nearly_optimal} and reflections~\cite{zindorf2024efficient,balauca2022efficient,zindorf_multicontrol_2025}.

\medskip
\textit{Analysis of DB-QITE.---}Similarly to ITE, DB-QITE obeys a fluctuation-refrigeration relation where the energy reduction in every step is bounded by the energy fluctuations as in Eq.~\eqref{eq:fluctuationRefrigeration}. 

\begin{theorem}\label{th: fluctuation-refrigeration main} \textbf{Fluctuation-refrigeration relation.}
The average energy $E_k:=\bra{\omega_k}\h\ket{\omega_k}$ of the states $\ket{\omega_k}:= U_k \ket 0$, where $U_k$ is defined recursively in Eq.~\eqref{DB-QITE Uk}, obeys 
\begin{align}
     \frac{E_{k+1} -  E_k}{s_k} \le - 2V_k + \mathcal O (s_k)
    \label{eq fluctuation-refrigeration main Os2} \ 
\end{align}
where $V_k:=\bra{\omega_k}\h^2\ket{\omega_k}-E_k^2$ is the variance of the energy in state  $\ket{\omega_k}$. 
\end{theorem}
We prove this statement in Theorem \ref{thm:QITE_DBQA_FRR} in App. \ref{sec: DBQA}.
This shows that in every step of DB-QITE  cooling rate matches the cooling rate of continuous ITE up to $\mathcal O(s_k^2)$.
Moreover, we give sufficient conditions on $s_k$ such that the higher order terms $\mathcal O (s_k^{2})$ do not overhaul the first order cooling.

The energy reduction of DB-QITE quantified in Eq.~\eqref{eq fluctuation-refrigeration main Os2} is warranted by non-negativity of the energy variances $V_k >0$.
However, Eq.~\eqref{eq fluctuation-refrigeration main Os2} does not yet exclude the possibility of converging to an excited energy eigenstate. 
Next, we will show that such scenarios do not occur.
Concretely, as proven in App.~\ref{sec: DBQA}, the following theorem holds. 

\begin{theorem}\label{th: fidelity convergence} \textbf{Ground-state fidelity increase guarantee.}
Suppose DB-QITE is initialized with a non-zero initial ground-state overlap $F_0$.
Let $\h$ be a Hamiltonian with a unique ground-state $| \lambda_0 \rangle$ and $\lambda_0 = 0$, spectral gap $\Delta$ and spectral radius $\|\h\|\ge 1$.
Let $U_0$ be an arbitrary unitary and set 
  \begin{align}
      s = \frac{\Delta}{12 \|\h\|^3}\ .
      \label{thm step duration}
  \end{align}
The states $\ket{\omega_k}:= U_k \ket 0$, where $U_k$ is defined recursively in Eq.~\eqref{DB-QITE Uk}, satisfy 
\begin{align} 
    F_{k+1} := |\langle \lambda_0|\omega_{k+1}\rangle|^2  \ge F_k \left(1 + \frac{(1-F_k) \Delta^2}{12 \|\h\|^3}\right) \geq 1- q^{k}\ 
    \label{thm main fk q}
\end{align}
where $q = 1- s  F_0 \Delta$.
\label{thmSimple}
\end{theorem}
This result shows that DB-QITE systematically synthesizes circuits that improve on previous circuits to prepare a better approximation to the ground-state. In particular, the first step is guaranteed to increase the fidelity with the ground-state by $F_0 (1-F_0) \Delta^2/ 12 \|\hat{H}\|^3$ where $F_0$ is the fidelity of the initial guess state $U_0 |0 \rangle $ and $\Delta$ is the spectral gap.
Moreover, subsequent iterations are guaranteed to further increase the fidelity to the ground-state and avoid converging to excited states (as could be consistent with Theorem~\ref{th: fluctuation-refrigeration main}). Hence DB-QITE provides a means of systematically preparing states with increased fidelity with the ground-state.

\medskip

\begin{figure*}
    \centering
    \includegraphics[width=1\linewidth, trim = 0 .8cm 0 .951cm ,clip
    ]{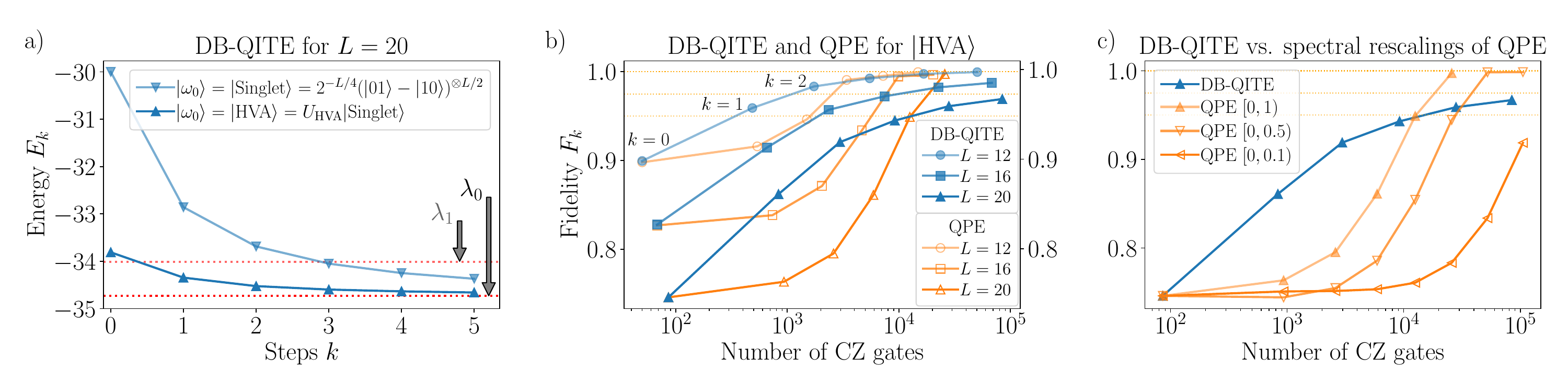}
    \caption{Numerical benchmarks for the 1d Heisenberg model. \textit{(a)} DB-QITE for $L=20$ qubits starting with either a product of singlets $\ket{\text{Singlet}}$ or a low-depth variationally learnt `warm-start' state $\ket{\text{HVA}}$ based on Refs.~\cite{bosse_Heisenberg_2022, kattemolle_vanwezel_2022,wecker2015progress}. 
    When extending $s_k$, just a handful of steps allows to reach $E_k\ll \lambda_1$.
    \textit{(b)} Two DB-QITE steps yield fidelities  $F_2\ge 95\%$ for $L=12$ and $F_2\ge 90\%$ for $L=20$ by using approximately $2\times 10^3$ CZ gates.
     QPE (orange) outperforms DB-QITE (blue) only when fault-tolerant execution allows for circuits with more than  $10^4$ CZ gates.
     In contrast, DB-QITE achieves high fidelity using fewer gates and more homogeneous circuits  (See App.~\ref{app Qrisp} for details).
    \textit{(c)}  QPE ideally uses a rescaled Hamiltonian $\h'= (\h-\lambda_0\1)/(\|\h\|-\lambda_0)$, but overestimating $|| \h||$ (e.g., by a factor of 2) significantly increases QPE runtime. 
    In such cases, DB-QITE can attain fidelity $F_4\approx 95\% $ more efficiently.     }
    \label{fig DBQITE QPE main} 
\end{figure*}

\textit{System-size scaling of DB-QITE.---}Resources required to run DB-QITE grow with system size but for a gapped local Hamiltonian the improvement from a single step of DB-QITE only degrades polynomially with system size $L$. 
Namely, it is clear from the formal statement of Theorem~\ref{th: fluctuation-refrigeration main} (i.e., Theorem \ref{thm:QITE_DBQA_FRR} in App.~\ref{sec: DBQA}) that a polynomially decreasing $s$ suffices to ensure the inequality in Eq.~\eqref{eq fluctuation-refrigeration main Os2} holds strictly and so for constant fluctuations $V_k$ the energy gain reduces only polynomially. 
Similarly, if $\Delta \in \Theta(1)$ and $\|\h\| \in \Theta(L)$ then the fidelity gain for a constant initial fidelity is only reduced by a polynomial factor. 
Moreover, the circuit depth required to implement a single iteration of DB-QITE only increases polynomially with system size (inline with the standard scaling of reflections and Hamiltonian simulation) and thus our circuit scalings for a single iteration of DB-QITE are  favorable. 
In the End Matter we further discuss the scaling for consecutive iterations of DB-QITE which are less favorable due to DB-QITE's recursive structure but next we provide numerical evidence that these scalings carry less weight when using larger step sizes $s_k$ than prescribed by worst-case bounds in Eq.~\eqref{thm step duration}.

\medskip

\textit{Numerical simulations of DB-QITE.---}We performed numerical simulations for the Heisenberg model with up to $L=20$ qubits and optimized $s_k$ by choosing the value giving the largest energy gain among $20$ evaluations of $E_k$ obtained with trial values for $s_k$.
Fig.~\ref{fig DBQITE QPE main}a) shows that this increases energy gain beyond Eq.~\eqref{thm step duration} and the linearization in Eq.~\eqref{eq fluctuation-refrigeration main Os2}; each DB-QITE step removes about half of the excess energy.
Quantitatively, DB-QITE reaches fidelity $F_2\ge 90\%$ for $k=2$ steps using less than $N_\text{CZ} \le 3\times 10^3$ CZ gates and $N_{U3}\le 4.8\times 10^3$ single-qubit gates.
Roughly half of this runtime suffices to achieve $F_2\ge 95\%$ for $L=12$ qubits.
These gate-counts are within reach of existing hardware~\cite{yamomotoPhaseEstimation,PhysRevLett.127.130505,Robledo_Moreno_2025}.
The circuit compilation is automated in the Qrisp programming language \cite{seidel2024qrisp} and the cost of implementing reflections is reduced by involving auxiliary qubits, see App.~\ref{app Qrisp} for details.
Qrisp allows us to compile circuits for $L\gg 20$ albeit without computing energy and fidelity gains.
If the same amount of Trotter-Suzuki steps in Hamiltonian simulation remains appropriate for $L=50$ then $k=1$ step of DB-QITE will require $N_\text{U3} \approx 3.4 \times 10^3$ single-qubit gates and $N_\text{CZ} \approx 2\times 10^3$ CZ gates while $k=2$ requires $N_\text{U3} \approx 1.2 \times 10^4$  and $N_\text{CZ} \approx 7\times 10^3$.

The main limitation of current large scale quantum devices is the fidelity of two-qubit gates, limiting the maximum realistically implementable CZ gates to $N_\text{CZ} \lesssim 10^3$~\cite{yamomotoPhaseEstimation}. 
Fig.~\ref{fig DBQITE QPE main}b) shows that if executed experimentally now, DB-QITE outperforms the canonical quantum phase estimation (QPE) variant reviewed in Ref.~\cite{ge2019faster}, a paradigmatic example of ground-state preparation algorithms geared towards fault-tolerant hardware~\cite{nearoptimalground_Lin_2020,fragmented_QITE2024,dong2022ground,epperly2021theory,temme2011quantum,kitaev1995quantum,brassard2002quantum,tan2020quantum,Low2019hamiltonian,PhysRevA.90.022305,double_bracket2024,PoulinWocjan,gilyen2019quantum,ge2019faster}.
Specifically, when $N_\text{CZ} \lesssim 10^3$,  DB-QITE yields $F_1\ge 95\%$ for $L=12$ and $F_2\ge 85\%$ for $L=20$ while QPE is limited to only 1 digit of precision and offers almost no improvement over initialization with $F_0\approx 75\%$. 

DB-QITE continues to outperform QPE when more CZ gates are available up until requiring fault-tolerance at around $N_\text{CZ}\gtrsim 10^4$.
This facilitates about $5$ digits of precision in QPE and better fidelities than DB-QITE.
However, it is likely a substantial resource underestimate that QPE would in practice reach such fidelities for $N_\text{CZ}\gtrsim 10^4$ because we had to make strong assumptions to achieve that: \textit{i)} knowing the energy of the ground-state and highest-excited state, \textit{ii)} all-to-all interactions and \textit{iii)} perfect measurement calibration.
In Fig.~\ref{fig DBQITE QPE main}b), we computed $\lambda_0$ and $\|\h\|$ by exact diagonalization for $L\le 20$ but in general these tasks are computationally intractable.
Fig.~\ref{fig DBQITE QPE main}c) highlights that 
even for $N_\text{CZ}\gtrsim 10^4$,
DB-QITE will continue to compete with QPE's performance: An overestimate of $\|\h\|$ by a factor 2 means that QPE needs $6$ auxiliary qubits to reach precision outperforming DB-QITE.
For a factor of $10$, we found that $7$ auxiliary qubits do not suffice.
DB-QITE's independence of rescalings is a general advantage as methods involving block-encodings rely on such rescalings too~\cite{nearoptimalground_Lin_2020,RapidInitializationPRXQuantum.6.020327,berry2024quantum}.

\medskip
\textit{Discussion.---}The fact that imaginary-time evolution (ITE) is a solution to double-bracket flows (DBFs) implies that ITE is an optimal gradient descent method in the sense of being a steepest-descent flow~\cite{optimization2012,smith1993geometric,moore1994numerical,bloch1990steepest,malvetti2024randomized}. 
The local optimality of DBFs has previously been investigated in the context of circuit compilation~\cite{riemannianflowPhysRevA.107.062421} and has motivated tangent-space methods crucial for
tensor network simulations~\cite{2011PhRvL.107g0601H,haegeman2014geometry,10.21468/SciPostPhysLectNotes.7,10.21468/SciPostPhys.10.2.040,10.21468/SciPostPhysCore.4.1.004, UnifyingTNSDBF,PhysRevResearch.5.033141}.
However, these studies involved classical processing at its core.
Here, rather than resorting to classical computations or heuristic hybrid quantum-classical variational methods~\cite{mcclean2018barren, larocca2024review, BittelKliesch,stilck2021limitations,bittel2022optimizing,cerezo2023does}, we took the radically different approach of implementing DBFs \textit{coherently} on the quantum computer~\cite{double_bracket2024}. 

Concretely, we propose the double-bracket quantum
imaginary-time evolution (DB-QITE) algorithm which exploits the connection of DBF and ITE to obtain unitaries in Eq.~\eqref{DB-QITE Uk} guaranteed to reduce the energy.
This is achieved through a recursive application of the group commutator.
DB-QITE involves forward and backward evolutions under the target Hamiltonian, interspersed with reflection operations and the initial state preparation unitaries. These ingredients are similar to bang-bang protocols for quantum control~\cite{lasalle1960bang,dong2010quantum,kuang2017rapid}.
However, while this quantum control paradigm is usually restricted to small systems, DB-QITE can be viewed as a quantum many-body echo protocol~\cite{tnsBangBangProtocolPhysRevB.106.195133}. 
Our  strategies can have implications for fields other than quantum computing which also employ ITE, e.g. path-integral optimization used to study holographic complexity in AdS/CFT correspondence~\cite{pathIntegeralOptimization,boruch2021holographic}.

DB-QITE inherits characteristics of ITE for ground-state preparation.
Specifically, Thm.~\ref{th: fidelity convergence} establishes that states prepared via DB-QITE are guaranteed to converge towards the ground-state  exponentially in the number of steps and Eq.~\eqref{eq fluctuation-refrigeration main Os2} shows that DB-QITE has, to first order, the same energy reduction rate as ITE. 
Both properties are ramifications of DB-QITE's rooting in gradient flow theory~\cite{optimization2012,smith1993geometric,moore1994numerical,bloch1990steepest,karimi2016linear}.
That the cooling rate of QITE is, similarly to ITE, determined by the energy fluctuations in the initial state is in line with the physical intuition that fast control of a system needs coherence in the energy basis~\cite{KosloffPRA,KosloffPRL}.
This suggests the possibility to derive quantum speed limits for cooling with a measurable metric that involves energy fluctuations.

While each step $k$ of DB-QITE is guaranteed to increase the ground-state fidelity, the circuit depth grows as $\mathcal{O}(3^k)$. 
Nonetheless, we provided numerical evidence that optimization of step durations $s_k$ leads to circuits within reach of existing quantum hardware and yielding high-fidelity approximations to ground-states of the Heisenberg models.
We found that for QPE to outperform DB-QITE, it required unrealistic assumptions, such as perfect knowledge of the spectrum and all-to-all connectivity.
We also found that when large circuit depths are available, then QPE's success probability is improved by initializing from a DB-QITE warm-start.
We expect DB-QITE warm-starts will be useful in conjunction with other methods too~\cite{nearoptimalground_Lin_2020,fragmented_QITE2024,dong2022ground,epperly2021theory,temme2011quantum,kitaev1995quantum,brassard2002quantum,tan2020quantum,Low2019hamiltonian,PhysRevA.90.022305,double_bracket2024,PoulinWocjan,gilyen2019quantum,ge2019faster,motta2020determining}. 
Our numerical results suggest that performing even a handful more DB-QITE steps is useful and this could be achieved by quantum dynamic programming~\cite{QDP}. 
In general, to keep the runtime low, it is essential to optimize  the step sizes $s_k$~\cite{double_bracket2024,xiaoyue2024strategies,robbiati2024double,thomson2024unravelling} which can be done based on measurements of the energy mean.
We conclude that DB-QITE, with its rigorous convergence guarantee and tractable gate counts, offers a refined pathway for ground-state preparation on large-scale quantum hardware.

\medskip
\textit{Code availability.} The source code to reproduce all numerical simulations in Qrisp and the figures is available openly at \href{https://github.com/renezander90/DB-QITE}{https://github.com/renezander90/DB-QITE}\ .

\medskip
\textit{Acknowledgments.}
We thank David Gosset for sharing his observation about QITE in Eq.~\eqref{eq:GML DBF} which he discovered while working on Ref.~\cite{AnshuImproved}.
Insightful discussions with  Jos\'e Ignacio Latorre and Alexander Jahn are gratefully acknowledged. 

MG, JS, BHT and NN are supported by the start-up grant of the Nanyang Assistant Professorship at the Nanyang Technological University in Singapore. MG has also been supported by the Presidential Postdoctoral Fellowship of the Nanyang Technological University in Singapore. ZH acknowledges support from the Sandoz Family Foundation-Monique de Meuron program for Academic Promotion.
RS and RZ are supported by European Union’s Horizon Europe research and innovation program under grant agreement no.\ 101119547.

%

\section*{End Matter}

\textit{Geometry of ITE DBF. ---} 
We show that the ITE dynamics in Eq.~\eqref{eq: DBF DE} represents the minimization of a cost function using its steepest gradient on the Riemannian manifold. 
To this end, we first show the fundamentals of DBFs and then elaborate on the property of ITE in terms of DBFs.

In general, DBFs are matrix-valued ordinary differential equations which have been shown to realize diagonalization, QR decompositions, sorting and other tasks~\cite{optimization2012,Brockett1991DBF,deift1983ordinary,Chu_iterations}.
Brockett first introduced the flow equation by studying the minimization of the least square of two matrices with the steepest descent techniques.
Let $\mathcal M(A) = \{U AU^\dagger\text{ s.t. } U^{-1}=U^\dagger\}$ be the set of all matrices generated by evolving a Hermitian operator $A$ under a unitary $U$. Consider the loss $\mathcal L_{A,B}(U)=-\frac{1}{2}\|U A U^\dagger-B\|^2_\text{HS}$ where $B$ is also Hermitian.
Then, the Riemannian gradient evaluated at $P=U AU^\dagger$ for the loss function $\mathcal L_{B}(P)=-\frac{1}{2}\|P-B\|_\text{HS}^2$ is given by~\cite{optimization2012,riemannianflowPhysRevA.107.062421}
\begin{align} \label{eq: grad_riemannian}
\text{grad}_P \mathcal{L}_B(P) = [[P,B],P]\ .
\end{align}
We note that the derivation relies on two conditions of Riemannian geometry, namely, the tangency condition and the compatibility condition, see Refs.~\cite{optimization2012,riemannianflowPhysRevA.107.062421} for a derivation.

This indicates that the the right-hand side of Eq.~\eqref{eq: grad_riemannian} is indeed the steepest descent direction of the cost function on the Riemannian manifold $\mathcal M(A)$. 
As the steepest-descent flow is the set of points $A(\ell)\in \mathcal M(A)$ which satisfy
\begin{align}
    \partial A(t)/\partial t = \text{grad}_{A(t)}\mathcal L_B(A(t))\ ,
\end{align}
we arrive at the flow equation
\begin{align}\label{eq: BrockettODE}
   \partial A(t)/\partial t = [[A(t), B],A(t)]\ ,\ 
\end{align}
which is exactly the DBF.

In case of the ground-state preparation task, we obtain the ITE in Eq.~\eqref{eq: DBF DE} by setting $A=\Psi(\tau)=\ketbra{\Psi(\tau)}{\Psi(\tau)}$ and $B=\h$.
Namely, the ITE realizes the minimization of the cost function
\begin{align} \label{eq: DBF cost}
    \mathcal{L}_{\h}(\Psi(\tau)) = - \frac12 \|\Psi(\tau) - \h\|_\text{HS}^2\ ,\
\end{align}
using its steepest descent direction on the Riemannian manifold.
Eq.~\eqref{eq: DBF cost} can be simplified to 
\begin{align} \label{eq: DBF cost_}
     \mathcal{L}_{\h}(\Psi(\tau))  = E(\tau) - \frac 12(1+ \|\h\|^2_\text{HS})\ .\
\end{align}
As the second term is independent of $\tau$, we can clearly verify that the cost function considered in this task is equivalent to the minimization of the energy, i.e., the preparation of the ground-state. Combining this observation with Eq.~\eqref{eq: BrockettODE} we obtain Eq.~\ref{eq: DBF DE}. That is, we see that DBFs provide a solution to ITE, as claimed. 
We remark that the ITE and Brockett's DBF have been extensively studied in slightly different fields, but their link has not been pointed out to the best of our knowledge.
However, individually they have been linked to gradient flows, see Ref.~\cite{HacklGeometry} for a discussion of ITE as gradient flow and Refs.~\cite{moore1994numerical,bloch1990steepest,smith1993geometric,dirr2008lie,kurniawan2012controllability,schulte2011optimal,schulte2008gradient} for a discussion of DBFs in context of gradient flows.

We also demonstrate some properties of QITE that can be inferred through the lens of DBF. For instance, the convergence to the ground-state can be checked by utilizing the properties of the DBF flow.
The solution of the DBF in Eq.~\eqref{eq: BrockettODE} is known to converge to the equilibria point that is characterized by $[A(\infty),B]=0$; that is, $A(\infty)$ shares the same eigenbasis with $B$.
Also, the DBF has the isospectral property, meaning any solution of the DBF has the same set of eigenvalues.
In case of the setting in Eq.~\eqref{eq: DBF cost}, as the initial state $\Psi(0)$ is pure (i.e., a rank-one matrix) and the cost function in Eq.~\eqref{eq: DBF cost_} indicates the energy minimization, the solution always converges to only one non-trivial eigenstate, i.e., the ground-state.

One can also derive the energy-fluctuation refrigeration relation for ITE in Eq.~\eqref{eq:fluctuationRefrigeration} directly from the link with DBFs. Concretely, by taking the derivative of Eq.~\eqref{eq:QITE}, we can see that the energy of ITE is driven by the energy fluctuation. 
As for the DBF, the derivative of the cost function in Eq.~\eqref{eq: DBF cost} reads
\begin{align}
    \partial \mathcal{L}_{\h} (\Psi(\tau))/\partial \tau = -\|[\Psi(\tau), \h]\|_\text{HS}
\end{align}
and further calculation reveals that
\begin{align}
 \|[\Psi(\tau), \h]\|_\text{HS} = 2V(\tau).\ 
\end{align}
This indicates that the energy dynamics of both cases are the same.
See~\ref{sec: DBF cooling rate} of Supplemental Materials for more discussion.

\medskip

\textit{Circuit depths for iterative applications of DB-QITE.---}
Theorem \ref{th: fidelity convergence} ensures that the fidelity with the ground-state converges exponentially with the number of DBQA \textit{steps} $k$. However, the depth of circuit (i.e., the number of queries to the Hamiltonian simulation, to reflections around the reference state or to initial state preparation circuit) also scales exponentially in $k$. 
Specifically, Eq.~\eqref{DB-QITE Uk} reveals that $k$ steps of DB-QITE require $\mathcal O(3^k)$  subroutine queries to Hamiltonian simulation or  partial reflections around the reference state.

The reflection unitary is a multi-qubit controlled parameterized phase gate, assuming the quantum state is the tensor-product zero state.
Multi-qubit controlled gates can be implemented efficiently without qubit overheads~\cite{barenco1995elementary}. 
Indeed Theorem~1 of Ref.~\cite{zindorf2024efficient} proves that any multi-qubit controlled unitary gate with $m\ge6$ controls is realizable by $12L-32$ CNOT gates with depth $8L-8$, $16L-48$ $T$ gates with depth $8L-6$ and $8L-32$ Hadamard gates with depth $4L-11$.
The App.~\ref{app Qrisp} describes the implementation in Qrisp which uses additional insights to further reduce these scalings.
This suggests that the gate complexity is linear in the number of qubits which is similar to efficient Hamiltonian simulation compiling~\cite{childs2010limitations,kothari,childs2012hamiltonian,PhysRevX.TrotterSuzukiError,SuChilds_nearly_optimal}.

Given the exponential query complexity in $k$, the DB-QITE analyzed in Theorem~\ref{th: fidelity convergence}
amplifies the initial fidelity $F_0$ to a prescribed final fidelity $F_\text{th}$ in depth 
\begin{equation}\label{eq: depth scaling}
\mathcal O\left(L\right)\times O\left(\left(\frac{1-F_0}{1-F_\text{th}}\right)^{24 \|\h\|^3/(  F_0\Delta^2)}\right)  \, ,
\end{equation}
where we multiplied the depth upper bound of the subroutines $\mathcal O(L)$ by the exponential query complexity, see Sec.~\ref{sec: DBQA} for derivation.
Thus, the depth scales exponentially with the inverse spectral gap $\Delta$ and with the inverse of the initial ground-state fidelity $F_0$.
The base of the exponential scaling is the ratio of the initial and final infidelity.
The last factor in the exponent is the inverse dependence on the step duration given in Eq.~\eqref{thm step duration}.
For local Hamiltonians this step duration is polynomially decreasing in the number of qubits $L$ implying that the runtime grows exponentially with $L^3$.
Thus, the  DB-QITE scheduling involved in the rigorous analysis allows for only a small number of steps $k$  so that circuit depths are modest.

There are strong indications that this runtime analysis is unnecessarily pessimistic.
In order to prove the Theorem we needed to use bounds that facilitated multiplicative rather than additive relations between the infidelities of consecutive DB-QITE states.
This is rather intricate and the bounds can be expected to not be tight.
More specifically, Eq.~\eqref{thm step duration} arises from taking a coarse lower bound $E_k\ge \Delta (1-F_k)$ which relates the energy to the fidelity $F_k$.
This lower bound is saturated if the $k$'th DB-QITE state $\ket{\omega_k}$ is supported only on the the ground-state $\ket{\lambda_0}$ and the first excited state $\ket{\lambda_1}$.
For such superposition of the two unknown lowest eigenstates, we can replace the Hamiltonian in our analysis by $\h_\text{eff} = \Delta\ket{\lambda_1}\bra{\lambda_1}$.
If $\h$ is gapped then the scheduling Eq.~\eqref{thm step duration} would be independent of the system size and the runtime would be system-size independent.

In turn, when the bound is not saturated then knowledge of $E_k$ (which is a basic measurement for monitoring any ground-state preparation quantum algorithm) would allow to choose a longer step duration and hence gain a larger energy decrease~\cite{double_bracket2024,xiaoyue2024strategies}.
For an unconstrained step duration $s_k$ in Eq.~\eqref{eq: GCI state main} an intermediate relation in the proof states
\begin{align}
\label{additive fidelity increase}
    F_{k+1} = F_k + 2 E_k s_k + \mathcal O(s_k^2) \ .
\end{align}
The magnitude of the higher-order terms determines the maximal step duration $s_k$. 
Rather than taking worst case estimates as we had to do in the proof, the higher-order terms can be estimated from simple measurements of the energy $E_k$ as a function of $s_k$ in Eq.~\eqref{eq: GCI state main}, see Ref.~\cite{xiaoyue2024strategies} for numerical examples.
In the proof, we upper bound the higher-order terms by the norm of the Hamiltonian but this is very likely a big over-estimate.
Indeed, those upper bounds are saturated when $\ket{\omega_k}$ is a superposition of the lowest and highest eigestates, which is highly unlikely.
Indeed, there is strong numerical evidence that in practice rather long steps can be used~\cite{double_bracket2024,xiaoyue2024strategies,robbiati2024double,thomson2024unravelling, motta2020determining}.
This makes it plausible that DB-QITE can be scheduled to gain much more fidelity in every step than is guaranteed by the worst-case lower bound Eq.~\eqref{thm main fk q}  which implies a much shorter circuit depth in those cases.

\textit{Comparison with prior work.---}Ref.~\cite{AnshuImproved} proposed circuit synthesis which provides an energy gain over product states in amount decided by energy fluctuations. 
This is similar to Eq.~\eqref{eq fluctuation-refrigeration main Os2} so we can reinterpret the energy reduction in Ref.~\cite{AnshuImproved} as being close to steepest-descent flows.
Ref.~\cite{AnshuImproved} improves the energy only for product states and hence cannot be iterated. 
In contrast, 
Thms.~\ref{th: fluctuation-refrigeration main} and~\ref{th: fidelity convergence} apply to arbitrary state initializations. Thus we not only amplify energy reductions with similarly  short depth circuits as Ref.~\cite{AnshuImproved}, but also allow the energy of arbitrary initial states to be reduced which allows to iterate our circuit synthesis. 
By this, the recursive character of DB-QITE should be viewed as a bonus despite the overheads associated with circuit repetitions.

Ref.~\cite{riemannianflowPhysRevA.107.062421} also exploits the commutator form of Riemannian gradients and derived a method for optimizing quantum circuits.
Their proposal aims to implement Eq.~\eqref{eq: DBI state main} by learning from measurements the parameters of the unitary. 
The protocol has rigorous convergence guarantees only in the limit of performing an exponential number of measurements and could be prohibitively expensive.

\textit{Quantum dynamic programming.---}
Let us comment how to reduce the circuit depth by extending its width.
Any state $\ket \psi = U_\psi \ket 0$ can be improved into a state with lowered energy by setting $\ket {\psi'} = U^{(\psi)} \ket \psi $
where
\begin{align}
    U^{(\psi)}=e^{i\sqrt{\tau}\h} e^{i\sqrt{\tau}\ket \psi\langle\psi|}    
    e^{-i\sqrt{\tau}\h}\ .
\end{align}
Thus in each step, DB-QITE depends on the input state's density matrix in a  non-trivial way and so it is a quantum recursion.
Recently, quantum dynamic programming has been proposed ~\cite{QDP} which similarly to ordinary dynamic programming leverages memoization to reduce the runtime of recursions.
In this case it suffices to combine density matrix exponentiation for the reflector $e^{i\sqrt{\tau}\ket \psi\langle\psi|}$ with regular Hamiltonian simulation.
See.~\cite{fragmented_QITE2024,nearoptimalground_Lin_2020} for discussion how to perform QITE using qubitization.

The reflection operator can be implemented using the density matrix exponentiation of $\ket{0}\bra{0}$.
Here, the cost scalings of the circuit depth and the number of copies of quantum states are respectively $\mathcal{O}(n\theta^2/\epsilon)$ and $\mathcal{O}(\theta^2/\epsilon)$ for implementing $e^{-i\theta\rho}$~\cite{lloyd2014quantum,Kimmel2017DME_OP}. 
While the precise realization could be costly, performing density-matrix exponentiation also seems feasible because the $\ket 0$ is known and so each time a copy is needed it can be prepared through a reset operation~\cite{kjaergaard2022demonstration}.

We stress that $U^{(\psi)}$ can be implemented obliviously to the circuit that prepares $\ket \psi$.
Thus, DB-QITE implemented through dynamic programming could be viewed a \textit{distillation} protocol of states with lowered energy.
The inverse of the number of quantum states needed to implement one step of density matrix exponentiation is a lower bound to the rate of approximate distillation.
If we want the distillation to have an increased gain we can perform more steps but with the trade-off that the rate will decay exponentially.

\title{Supplemental Material for ``Double-bracket quantum algorithms for quantum imaginary-time evolution"}
\maketitle
\onecolumngrid

\tableofcontents

\newpage
\section{Common terminology and techniques}\label{sec: terminology}

In this section, we make a comprehensive list of common terminology used throughout the rest of the appendix. 
\begin{enumerate}

\item The \emph{Hilbert-Schmidt norm} (or \emph{Frobenius norm}) of an operator \( A \) acting on a Hilbert space \( \mathcal{H} \) is defined as:
\begin{align}
    \| A -B\|_{\text{HS}} = \sqrt{\text{Tr}((A-B)^\dagger (A-B))}.
\end{align}

\item The \emph{operator norm} \( \| \hat{X} \| \) is defined as the smallest number such that for all normalized vectors, \( \| \hat{X} \psi \| \leq \| \hat{X} \| \) holds:
\begin{align}
    \| \hat{X} \| = \sup_{\| \psi \| = 1} \| \hat{X} \psi \|.
\end{align} 
When the operator \( \hat{X} \) is Hermitian on a finite-dimensional inner product space, its operator norm is equal to the maximum absolute value of its eigenvalues.
Any norm satisfies the \emph{triangle inequality}
$    \| x + y \| \leq \| x \| + \| y \|$.

\item The Taylor series expansion of a function \( f(x) \) around \( x_0 \), truncated at the \( n \)-th term, is
\begin{align}
    f(x) = f(x_0) + f^{(1)}(x_0)(x - x_0) + \cdots + \frac{f^{(n-1)}(x_0)}{(n-1)!}(x - x_0)^{n-1} +\frac{1}{n!} (x - x_0)^n f^{(n)}(\xi)\ .
\end{align}
where \( f^{(n+1)}(t) \) is the \((n+1)\)-th derivative of \( f \) and $\xi \in[x_0,x]$.
\item For two pure quantum states \( |\psi\rangle \) and \( |\Omega \rangle \), the fidelity is defined as
$ F(|\psi\rangle, |\Omega \rangle) = |\langle \psi |\Omega \rangle|^2$.
To analyze fidelity convergence in QITE, it is useful to derive bounds on the expected energy and energy variance in terms of the ground-state infidelity, $\epsilon=1-F$ and the spectrum of eigenenergies $\{\lambda_j\}_j$.

\begin{minipage}{0.96\textwidth}
\begin{lemma}\label{lemma: Ek > lambda epsilon}
Let $| \Omega \rangle$ be any pure state. 
Denote Hamiltonian as $ \h= \sum_{j=0}^{d-1} \lambda_j |\lambda_j \rangle \langle \lambda_j |$ where we assume that the eigenvalues of $\hat H$ are ordered increasingly and we set $\lambda_0=0$.
Suppose that the ground-state fidelity is given by $F=1-\epsilon$, then we have
  \begin{align}
      E_k\ge \lambda_1\epsilon \ , \quad \text{ and } \quad V_k\le \lambda_{d-1}^2\epsilon \ ,
  \end{align}
  where $E_k=\langle \Omega  |\h^2|\Omega \rangle $ is the expected energy and $V_k =  \langle\Omega |\h^2|\Omega  \rangle - E_k^2$ is the energy variance.
\end{lemma}
\begin{proof}
Let's define the probability to occupy the $i$-th eigenstate $| \lambda_i \rangle$ as $p_i$.
    As the ground-state fidelity is given by  $F=1-\epsilon$, the state has probability $p_0= 1-\epsilon$ to occupy the ground-state.
    Since $\lambda_0=0$ by assumption, we obtain 
     \begin{align}\label{eq: Ek > lambda epsilon}
   E_k =\sum_{i=1}^{d-1}  \lambda_i p_i \geq \lambda_1 \sum_{i=1}^{d-1}   p_i = \lambda_1\epsilon \ , \quad  \text{ and } \quad V_k \leq \sum_{i=1}^{d-1}  \lambda_i^2 p_i \le \|\h\|^2 \sum_{i=1}^{d-1}  p_i = \|\h\|^2 \epsilon_k \ .
      \end{align}
\end{proof}
\end{minipage}

\renewcommand{\arraystretch}{1.5} 
\begin{table}[h!]
\begin{tblr}{
    colspec = {|c|c|c|c|},
    row{2} = {blue9},
    row{4} = {blue9},
    row{6} = {blue9},
    row{8} = {blue9},
    row{10} = {blue9},
  }
\hline
 & \makecell{Sec.~(\ref{sec: QITE}): ITE} & \makecell{Sec.~(\ref{sec: QITE DBI}): QITE DBI} & \makecell{Sec.~(\ref{sec: DBQA}): DB-QITE}\\ 
\hline
Initial state   & $| \Psi_0 \rangle $    & $ | \sigma_0 \rangle = | \Psi_0 \rangle$           & $ | \omega_0 \rangle =| \Psi_0 \rangle$                                                                     \\
Evolution time step    & Continuous duration $\tau$    & Discrete duration $s_k$    & Discrete duration $s_k$                                                                    \\
Evolved states & $|\Psi(\tau) \rangle$, see Eq.~\eqref{eq:Psitau} & $ | \sigma_k \rangle $, see Eq.~\eqref{eq: DBI definition}                  & $ | \omega_k \rangle $, see Eq.~\eqref{eq:reducedGCI}             
\\
Expected energy   & $E(\tau) =\langle \Psi(\tau) | \h | \Psi(\tau) \rangle$   & $\overline{E}_k = \langle \sigma_k|\h|\sigma_k\rangle$        & $E_k = \langle \omega_k|\h|\omega_k\rangle$                                                \\[2pt]
Energy variance   & $V(\tau)= \langle \Psi(\tau) | \h^2 | \Psi(\tau) \rangle-E(\tau)^2$ & $\overline{V}_k =  \langle \sigma_k |\h^2|\sigma_k \rangle - \overline{E}_k^2$ & $V_k =  \langle \omega_k |\h^2|\omega_k \rangle - E_k^2$ 
\\[3pt]
\makecell{Fluctuation-refrigeration\\relation} & $\partial_\tau E(\tau) = - 2 V(\tau)$ & $\overline{\E}_{k+1} \le  \overline{\E_k} - 2s_k \overline{V}_k + \mathcal O (s_k^{2})$& $\E_{k+1} \le  \E_k - 2s_k V_k + \mathcal O (s_k^{2})$
\\[2pt]
\makecell{Time step conditions for\\ $E_{k+1}-E_k \leq -s_k V_k$}  & \makecell{N.A.} & \makecell{$ \displaystyle s_k \leq \dfrac{\V_k}{4 \|\h\| \langle \sigma_k| \h^2 |\sigma_k\rangle } $\\[10pt] (Proposition~\ref{prop: DBI frr}) }& \makecell{$\displaystyle s_k  \le \frac{4V_{k}}{5 \epsilon_{k}\lVert \h\rVert^{4}}$\\[10pt](Theorem~\ref{thm:QITE_DBQA_FRR})}
\\[4pt]
ground-state fidelity & $F\left(\ket{\Psi({\tau})},\ket{\lambda_0}\right) = \left|\braket{\lambda_0}{\Psi(\tau)}\right|^2$  & $F_k = | \langle \lambda_0 | \sigma_k \rangle|^2=1 - \epsilon_{k}$                                   & $F_k = | \langle \lambda_0 | \omega_k \rangle|^2=1 - \epsilon_{k}$                        \\[5pt]
\makecell{ground-state\\fidelity convergence} &   $F(\ket{\Psi({\tau})},\ket{\lambda_0})\geq 1-\delta_H$ 
& \makecell{$F_k \ge 1- q^{k} $ \\ with $q=1- sF_0 \lambda_1 $} 
& \makecell{$F_k \ge 1- q^{k} $\\
with  $q=1-sF_0 \Delta$}
\\[2pt]
\makecell{Time step conditions\\for the above fidelity\\convergence} & \makecell{$\displaystyle \tau_H = \log \left[O\left(\frac{L}{\delta_H |\langle \lambda_0| \Psi_0 \rangle|^{2})} \right) \right]$\\
(Sec.~\ref{subsec:QITE_convergence})}
& \makecell{ $\displaystyle s=\frac{\sqrt{F_0} }{4\| \h \|} $\\
(Theorem~\ref{th: dbi fidelity revised})}
& \makecell{ $\displaystyle s = \frac{\Delta}{12 \|\h\|^3}$\\
(Theorem~(\ref{th: fidelity convergence}) of main text \\ and Sec.~\ref{app:subsec:convergence})}\\[2pt]
\hline
\end{tblr}
\caption{The summary of notations for imaginary-time evolution (ITE), quantum imaginary-time evolution with double-bracket flows (QITE DBI), and with the explicit compiliation using techniques in double-bracket quantum algorithms (DB-QITE), for ease of reference.
For all sections of the supplemental materials, we define the Hamiltonian as $ \displaystyle \h= \sum_{j=1}^{d-1} \lambda_j |\lambda_j \rangle \langle \lambda_j |$ where $d$ is the dimension of the Hamiltonian. We also denote the spectral gap as $\Delta= \lambda_1- \lambda_0$ where $\lambda_1$ and $\lambda_0$ represent the first excited and ground-state energy respectively. We also define $\| \h\|$ as the operator norm of the Hamiltonian. Finally, we assume that the initial state has some non-zero ground-state overlap $F_0 > 0$.}
\end{table}

\end{enumerate}

\newpage
\section{Key properties of  imaginary-time evolution (ITE)}\label{sec: QITE}
This is a section that outlines all the key properties of ITE which are relevant for the scope of our work. We begin with the key observation in Section \ref{sec:QITE is DBFQITE as a double-braket flow} that the state defined via ITE
is a solution to the double-bracket flow equation. This leads to the decrease in average energy of the ITE state with respect to time, indicating convergence towards the ground-state. 
In Section \ref{sec: DBF cooling rate}, we derive a fluctuation-refrigeration relation which shows that states with higher energy fluctuations have a higher cooling rate. 
We next provide a digression in Section \ref{subsec:thermal_convergence} discussing convergence of thermal states to ground-states.
We then carry over the techniques used in this section, to present in Section \ref{subsec:QITE_convergence} an exponential convergence of the ITE state's fidelity to the ground-state. 
This means that implementation of ITE DBF Eq.~\eqref{eq app QITE DBF}  as a double-bracket quantum algorithm~\cite{double_bracket2024} will converge to the ground-state too.

\subsection{Proof that  imaginary-time evolution is a solution to a double-bracket flow}\label{sec:QITE is DBFQITE as a double-braket flow}

Recall the definition for $\tau>0$, and an initial pure state vector $\ket{\Psi_0}$ we set
\begin{align}\label{eq:Psitau}
	|\Psi(\tau) \rangle = \frac{e^{-\tau \hat H} |\Psi_0\rangle}  {\| {e^{-\tau \hat H}}|\Psi_0\rangle \| } \ ,
\end{align}
where $\|\ket \psi\|=\sqrt{\braket{\psi}\psi}$ for any vector $\ket \psi$. 
Next, we prove a key connection between ITE and DBFs.
\begin{proposition}[ITE is a solution to a DBF] \label{prop :QITE is DBFQITE as a double-braket flow}
 The ITE pure state $\hat\Psi(\tau)=|\Psi(\tau) \rangle\langle\Psi(\tau) |$ is a unique solution to Brockett’s double-bracket flow equation, i.e.
 \begin{align}
     \frac{\partial \Psi(\tau)}{\partial \tau} = \big[ [\Psi(\tau),\h],\Psi(\tau)\big]\ .
     \label{eq app QITE DBF}
 \end{align}
We refer to Eq.~\eqref{eq app QITE DBF} as ITE DBF.
\end{proposition}
\begin{proof}
First, observe that the ITE map Eq.~\eqref{eq:QITE} is smooth as a function of $\tau$.
Second, ITE preserves the state vector normalisation so there must exist a unitary $U_\tau$ that implements it.
Smoothness implies that there is a unique unitary $U_\tau$ such that $|\Psi(\tau) \rangle = U_\tau |\Psi_0 \rangle \ $.
We begin the proof by denoting the normalization factor of ITE as
\begin{align*}
N(\tau) = \|e^{-\tau \h} |\Psi_0 \rangle\| = \sqrt{\langle \Psi_0 | e^{-2\tau \h} | \Psi_0 \rangle}, 
\label{normalization_modified}
\end{align*}
and evaluating its derivative to be 
\begin{align}
    \partial_\tau \left( N(\tau)\right ) = \frac{-\langle \Psi_0 | \h e^{-2\tau \h}| \Psi_0\rangle}{N(\tau)} \ .
\end{align}
Thus the change of the ITE state vector as a function of $\tau$ is~\cite{GellmannLow}
\begin{align}
    \partial_\tau {|\Psi(\tau) \rangle} 
    = \partial_\tau \left( \frac{e^{-\tau \h} {| \Psi_0 \rangle}}{N(\tau)} \right ) &= \frac{ -\h e^{-\tau \h} {| \Psi_0 \rangle} }{N(\tau)} + \frac{ e^{-\tau \h} {| \Psi_0 \rangle} \cdot \left( - \partial_\tau N(\tau) \right ) }{N(\tau)^2} \\
    &= -\h |\Psi(\tau) \rangle + \frac{ \langle \Psi_0| \h e^{-2\tau \h} |\Psi_0\rangle }{N(\tau)^2 } |\Psi(\tau) \rangle  \\
    &= -\h |\Psi(\tau) \rangle + \langle \Psi(\tau) | \h | \Psi(\tau) \rangle |\Psi(\tau) \rangle\ .
\label{eq:QITEstandard}
\end{align}
As a remark, this is of the form of a nonlinear Schr\"odinger equation. On the other hand, we note that the ITE density matrix also satisfies
$    \Psi(\tau) = |\Psi(\tau) \rangle \langle \Psi(\tau) | \equiv  U_\tau \Psi_0  U_\tau^{\dagger}$
satisfies
\begin{align}
    [\Psi(\tau)  ,\h]|\Psi(\tau) \rangle &= |\Psi(\tau) \rangle \langle \Psi(\tau) | \h |\Psi(\tau) \rangle  -  \h|\Psi(\tau) \rangle \langle \Psi(\tau)  |\Psi(\tau) \rangle =  -\h |\Psi(\tau) \rangle + E(\tau) |\Psi(\tau) \rangle\ .
\end{align}
Thus, from Eq.~\eqref{eq:QITEstandard} we have $ \partial_\tau {|\Psi(\tau) \rangle} = [\Psi(\tau)  ,\h]|\Psi(\tau) \rangle$.
To complete the proof we set $\mathcal W = i[\Psi(\tau)  ,\h]$ and obtain Eq.~\eqref{eq: DBF DE} as the von Neumann evolution equation of the ITE density matrix.

\end{proof}

\begin{remark}
      Note that Eq.~\eqref{eq: DBF DE} is a Brockett DBF, see discussion in End Matter.
Let us discuss the physical meaning of this ITE state.
First, the expected energy of $\Psi(\tau)$ will keep decreasing until it reaches some fixed point $\Psi_\infty$.
Clearly, the condition for a fixed point is $[\Psi_\infty,\h]=0$, i.e. both of them are diagoanalizable in the energy basis.
Thus, they can be expressed as
\begin{align}
    \h= \sum_{j=1}^{d-1} \lambda_j |\lambda_j \rangle \langle \lambda_j |\;\; \text{ and } \;\;|\Psi(\infty) \rangle = \sum_{j=0}^{d-1} c_j |\lambda_j \rangle \langle \lambda_j | \ .
\end{align}
Without loss of generality, we assume that the eigenvalues of $\hat H$ are ordered increasingly, i.e. $\lambda_n \leq \lambda_{n+1}$.
Note that $\Psi(\tau)$ is a pure state, and any unitary conjugation preserves rank. 
Additionally, since $\Psi_\infty$ is diagonalizable in the energy basis, this means that it is an energy eigenstate.
This means that for values of $\tau >0$, the ITE state converges to the ground-state $\Psi_\infty =  |\lambda_0 \rangle \langle \lambda_0 | $.
On the other hand, if $\tau <0$, then the ITE state converges to $\Psi_\infty =  |\lambda_{d-1} \rangle \langle \lambda_{d-1} | $ which is the highest energy state of $\h$.

\end{remark}

\subsection{Fluctuation-refrigeration relation for the ITE cooling rate }\label{sec: DBF cooling rate}
In this section, our aim is to derive a \emph{fluctuation-refrigeration relation} for DBF, which shows that states with high energy fluctuation will converge to the ground-state more quickly by ITE procedure. We have seen that in the continuous-time DBF formalism, the average energy in the ITE state vector should decrease when $\tau$ increases as it converges to the ground-state, i.e. 
\begin{align}
E(\tau) := \langle \Psi(\tau) | \h | \Psi(\tau) \rangle
\end{align}
is a decreasing function of $\tau$. We show that the decrease rate is dictated by the energy variance $V(\tau) := \langle \Psi(\tau) | \h^2 | \Psi(\tau) \rangle-E(\tau)^2$.

\begin{proposition}[ITE Fluctuation-refrigeration relation]\label{prop:frr_QITE}
 The derivative of average energy with respect to inverse temperature, $\partial_\tau E(\tau) $, is given by 
 \begin{align}
     \partial_\tau E(\tau) = - 2 V(\tau) \ .
 \end{align}
\end{proposition}

\begin{proof}
We may evaluate the derivative directly using Eq.~\eqref{eq:QITEstandard} and the Leibniz chain rule, 
\begin{align}
    \partial_\tau E(\tau) &=  \langle \Psi(\tau) | ( E(\tau)-\h)  \h| \Psi(\tau) \rangle + \langle \Psi(\tau) | \h( E(\tau)-\h)  | \Psi(\tau) \rangle\ .
\end{align}
Combining the two terms give $\partial_\tau E(\tau)
       = -  2\langle \Psi(\tau) | \h^2 | \Psi(\tau) \rangle +2E(\tau)\langle \Psi(\tau) | \h | \Psi(\tau) \rangle = -2 V(\tau)$.
\end{proof}

\begin{remark}
    To understand the limits of the relation let us look at asymptotic upper bounds.
We will find that they can be attained for `hot' states, where energy fluctuations scale in the system size.
On the other hand, for states close to eigenstates energy fluctuations vanish.
For a general quantum system, the energy fluctuation is upper bounded by 
\begin{align}
    V(\tau) &\leq \langle \Psi(\tau) | \h^2 | \Psi(\tau) \rangle \leq \| \h\|^2 \ .
    \end{align}
A stronger bound is obtained for local Hamiltonians, e.g. if they involve $I$ terms  $\h = \sum_{i=1}^I \h_i$ where each $\|\h_i\|\le \mathcal{O}(1)$ has norm independent of system size. In this case, for a system containing $L$ qubits,
\begin{align}
    E(\tau) \leq \mathcal{O}(L), \qquad V(\tau) &\leq \mathcal{O}(L^2)\ .
\end{align}

For example for the transverse field Ising model $\h_\text{TFIM} = \sum_{i=1}^{L-1} (Z_iZ_{i+1}+X_i)$ we can set $\h_i = Z_iZ_{i+1}+X_i$ where $Z_i$ and $X_i$ act on qubit $i$ and we have $\|\h_i\| = 2$. 
For $ \ket{\Psi_0}= \ket{+}^{\otimes L}$ we have $E(0) = L-1 $ and  $V(0) = L-1$ because $(Z_iZ_{i+1})^2 =1$ and contribute to the energy fluctuation.

\end{remark}

\subsection{Digression: Exponential convergence of thermal states to ground-states in 1-norm}\label{subsec:thermal_convergence}

We will next recall a result by Hastings~\cite{hastingsBound} who showed that thermal states of gapped Hamiltonians converge exponentially fast to the ground-state if it is unique and the density of eigenstates above the gap is polynomial.
This is a digression because in our work we are dealing with pure states rather than thermal states which are generically mixed.
However, the discussion is instructive because the same assumption and same mathematical idea will allow us to show an analogous result for ITE. 
 
More specifically, in Ref.~\cite{hastingsBound}, Hastings  identified a condition on the density of states which suffices to obtain the scaling of the inverse temperature of a thermal state such that it approximates the unique ground-state of a gapped system.
More specifically, let $\hat H$ be a Hamiltonian with eigenvalues 
$\{\lambda_j\}$ and we have labeled by $ j=0,1,\ldots,\text{dim}(\h)-1$, again ordering the eigenvalues increasingly such that $\lambda_j \leq \lambda_{j+1}$. Furthermore, let us set the ground-state energy as $\lambda_0=0$. 
We can define a simple counting function that captures the density of energy eigenstates. In particular, given the spectral gap $\Delta = \lambda_1-\lambda_0$, let us define a series of intervals $\mathcal{M}_m  := [m\Delta , (m+1)\Delta)$. The function $\rho_m(\h) = \big|\{ \lambda_j \in\mathcal{M}_m \}\big|$, for $m = 0,1,2,\ldots$ therefore counts the number of eigenvalues in each interval $\mathcal{M}_m$.
With this definition, Hastings' polynomial density of states condition is expressed as
\begin{align}\label{eq:Hastings}
\rho_m(\h) \le  (c_H L)^m/ m!
\end{align}
for $c_H=O(1)$, $L$ being the number of qubits.
In other words, we assume that the distribution of eigenenergies is sparse near the ground-state for any gapped Hamiltonian.
This allows to bound for any inverse temperature $\beta$ the partition function
\begin{align}
Z_\beta = \sum_{j=0}^{\text{dim}(\h)-1} e^{-\beta \lambda_j} \le e^{-\beta \lambda_0} \left[\sum_{m=0}^\infty \rho_m(\h) e^{-\beta \Delta m} \right]
&\le e^{-\beta \lambda_0} \sum_{m=0}^\infty\frac{(e^{-\beta \Delta } c_H L)^m}{m!}  \\
 &=  e^{-\beta \lambda_0} \exp(e^{-\beta \Delta E } c_H L)
\end{align}
where in the first inequality, we use the fact that for all $\lambda_j \in \mathcal{M}_m$, we have that $\lambda_j \geq m\Delta $ and therefore $e^{-\beta E} \leq e^{-\beta \Delta m}$. To illustrate this further, if $\rho_0=\ketbra{\lambda_0}{\lambda_0}$ is the ground-state projector and $\rho^{(\beta)}=\exp(-\beta\h)/Z_\beta$ is the thermal state (as opposed to ITE state it is generically mixed and not pure), then
\begin{align}
\|\rho^{(\beta)}-\rho_0\| \le 2\sum_{m=1}^\infty \rho_m(\h) e^{-\beta \Delta E m}
\le2 \left[ \exp(e^{-\beta \Delta E } c_H L)-1 \right].
\end{align}
Thus, if we set  \vspace{-0.3cm}
\begin{align}
    \beta_H =\dfrac{1}{\Delta} \left[ \log\left(\dfrac{c_H L}{\log(1+\varepsilon/2)} \right) \right]\ ,
\end{align}
then $\|\rho^{(\beta_H)}-\rho_0\| \le \varepsilon\ $.
To understand this more intuitively we use $\log(1+\varepsilon/2)  \ge \varepsilon/4$ which implies $\beta_H\le \frac 1\Delta\log(\mathcal{O}( L \epsilon^{-1}))$.
This means that thermal states of gapped Hamiltonians converge to the respective ground-states in inverse temperature scaling logarithmically in the system size and the desired approximation precision.

\subsection{Exponential convergence of fidelity to the ground-state through ITE}\label{subsec:QITE_convergence}
In the previous section, we studied the exponential convergence of thermal states to the ground-state, under Hasting's condition on the system Hamiltonian. However, the central interest of this work revolves around ITE states. While such states are pure as opposed to generically mixed thermal states, both involve an imaginary-time exponential of the Hamiltonian and a normalization. This is a key feature that allows for similar fidelity convergence behavior on the ITE state, to the ground-state. 

To see this explicitly, we apply a similar consideration to the ground-state fidelity of the ITE which we define as
\begin{align}\label{eq:qite_fidelity}
F(\ket{\Psi({\tau})},\ket{\lambda_0}) = |\braket{\lambda_0}{\Psi(\tau)}|^2 =\frac { e^{-2\tau \lambda_0} |\braket{\lambda_0}{\Psi_0}|^2}{ \left\| e^{-\tau \hat H} \ket{\Psi_0}\right\|^2
}\ ,
\end{align}
where one should recall that $\ket{\Psi_0}$ is the initial state of ITE, see Eq.~\eqref{eq:Psitau}. 
A lower bound for Eq.~\eqref{eq:qite_fidelity} can be obtained by upper bounding the ITE normalization
\begin{align}
\|e^{-\tau \hat H} \ket{\Psi_0}\|^2 = \bra{\Psi_0}e^{-2\tau \hat H} \ket{\Psi_0}
&= \sum_{j=0}^d |\langle \lambda_j| \Psi_0 \rangle|^2 e^{-2\tau \lambda_j}\\ 
&= |\langle \lambda_0| \Psi_0 \rangle|^2 e^{-2\tau \lambda_0}+ \sum_{j=1}^d |\langle \lambda_j| \Psi_0 \rangle|^2 e^{-2\tau \lambda_j}\\ 
&\le e^{-2\tau \lambda_0} \left[ |\braket{\lambda_0}{\Psi_0}|^2 + \sum_{m=1}^\infty e^{-2\tau m\Delta} \sum_{j:\lambda_j\in \mathcal{M}_m}|\langle \lambda_j| \Psi_0 \rangle|^2\right] \\
&\le e^{-2\tau \lambda_0} \left[ |\braket{\lambda_0}{\Psi_0}|^2 + \sum_{m=1}^\infty \rho_m(\h) e^{-2\tau m\Delta}\right]
\label{eq:simp_uppbound}\\
&\leq  e^{-2\tau \lambda_0} \left( |\braket{\lambda_0}{\Psi_0}|^2 + \exp \left(e^{-2\tau \Delta} c_H L\right)-1\right)\ .
\end{align}
In Eq.~\eqref{eq:simp_uppbound} we use a very loose upper bound $|\braket{\lambda_j}{\Psi_0} \rangle| \leq 1$ to insert the density of states $\rho_m(\h)$, and in the last line we invoked Hastings' condition. Finally, we proceed to upper bound $F(\Psi_0,\lambda_0)$ in Eq.~\eqref{eq:qite_fidelity}:
\begin{align}\label{eq: fidelity for Hasting QITE}
    F(\ket{\Psi({\tau})},\ket{\lambda_0})\geq \frac{|\braket{\lambda_0}{\Psi_0}|^2}{|\braket{\lambda_0}{\Psi_0}|^2+\exp \left(e^{-2\tau \Delta} c_H L\right)-1} \geq \frac{1}{1+\delta_H} \geq 1-\delta_H,
\end{align}
where we have set $\displaystyle \delta_H = \frac{\exp \left(e^{-2\tau \Delta} c_H L\right)-1}{|\braket{\lambda_0}{\Psi_0}|^2}$. 
Inverting this relation we find that it suffices for imaginary-time evolution to have the duration given by
\begin{align}\label{eq: flow duration for HastingsQITE}
\tau_H = \frac 1{2\Delta} \log\bigg(\frac{c_H L}{\log(1+\delta_H |\langle\lambda_0|\Psi_0\rangle|^2)}\bigg)\ .
\end{align}
In Eq.~\eqref{eq: fidelity for Hasting QITE} we defined $\delta_H>0$ as the desired ground-state infidelity for ITE initialized with a non-zero ground-state overlap $F(\ket{\Psi_{0}},\ket{\lambda_0})=|\langle \lambda_0| \Psi_0 \rangle|^2\neq 0$.
If $\h$ has a unique gapped ground-state and it satisfies Hastings' density of states condition in Eq.~\eqref{eq:Hastings}, then the scaling 
\begin{align}
    \tau_H = \log(O(L \delta^{-1}_H |\langle \lambda_0| \Psi_0 \rangle|^{-2}))
\end{align}
in  Eq.~\eqref{eq: flow duration for HastingsQITE} suffices for ITE to achieve the desired high fidelity.

\section{Double-bracket iteration (DBI) approach to QITE} \label{sec: QITE DBI}
In the previous section, we established that QITE is a DBF. Importantly, we were able to show two key performance indicators of QITE: the cooling rate as a function of energy variance, and the fact that under Hastings' assumption, one may obtain an exponential convergence for the fidelity of the QITE state to target ground-state. Nevertheless, for both numerical and compilation purposes, it will be necessary to demonstrate these properties for \emph{discretizations} of the DBF method. We begin by a considering the discretization which we referred to as \emph{double-bracket iteration} (DBI) though previous works used the name Lie bracket recursions~\cite{optimization2012}.

For the computational purposes of this appendix, we introduce the following notation:
denote $\ket{\sigma_0}$ as the initial state fed into a discretized DBI computation, and $\lbrace\ket{\sigma_{k}}\rbrace_{k}$ as the sequence of states that solve the QITE DBI, which takes on a recursive form: 
\begin{align}\label{eq: DBI definition}
    \ket{\sigma_{k+1}} = e^{s_k[\sigma_k,\h]} \ket{\sigma_k}.
\end{align}
In the above, we have denoted the time step size in the $(k+1)$-th DBI iteration as $s_k$. Furthermore, to simplify the notation we use the density matrix representation $\sigma_k = | \sigma_k \rangle \langle \sigma_k|$.
The quantities of interest, as in Section \ref{sec: DBF cooling rate}, are given by the average and variance of energy after $k$-th DBI iteration:
\begin{align}\label{eq:qite_dbi_e&v}
    \overline{E}_k = \langle \sigma_k|\h|\sigma_k\rangle, \quad \text{and} \qquad  \overline{V}_k =  \langle \sigma_k |\h^2|\sigma_k \rangle - \langle \sigma_k|\h^2|\sigma_k\rangle^2.
\end{align}

\subsection{Useful lemmas for proving fluctuation-refrigeration relation and fidelity convergence of QITE DBI }

\begin{enumerate}

\item We will use the following lemma in Sec.~\ref{app: DBI cooling rate}.

\begin{minipage}{0.96\textwidth}
\begin{lemma} Let $\W = -\W^\dagger$ and $\h =\h^\dagger$. Then         \begin{align}\label{eq: nth derivative of H evolution}
         \partial_s^n  \; \bigg( e^{-s\W}\h e^{s\W} \bigg) = e^{s\W}[(\W)^n, \h] e^{-s\W} \ ,
     \end{align}
     where $n \in \mathbb{N} $ indicates the order of the derivative and we define the notation for nested commutators 
     \begin{align}
         [(X)^n, Y] = \bigg[X, [(X)^{n-1}, Y]  \bigg] \quad \text{ with } \quad[(X)^0, Y] = Y  \ ,
     \end{align}
\end{lemma}
\end{minipage}
\begin{minipage}{0.96\textwidth}
\begin{proof}
We prove it by mathematical induction. For the base case \( n = 1 \), we have 
\begin{align}
    \partial_s \left( e^{s\W} \h e^{-s\W} \right) = e^{s\W}\W \h e^{-s\W} - e^{s\W} \h\W e^{-s\W} = e^{s\W} [\W, \h ] e^{-s\W}
\end{align}
Next, we assume that for some integer \( n \geq 1 \):
\begin{align}
    \partial_s^n \left( e^{s\W} \h e^{-s\W} \right) = e^{s\W} [ (\W)^n, \h ] e^{-s\W}
\end{align}
The \( (n+1) \)-th derivative is then given by
\begin{align}
\partial_s^{n+1} \left( e^{s\W} \h e^{-s\W} \right) 
&= \partial_s \left( e^{s\W} [ (\W)^n, \h ] e^{-s\W} \right) \\
&= e^{s\W}\W [ (\W)^n, \h ] e^{-s\W} - e^{s\W} [ (\W)^n, \h ]\W e^{-s\W} \\
&= e^{s\W} [\W, [ (\W)^n, \h ] ] e^{-s\W} .
\end{align}
Using the definition of the nested commutator, we have $    [\W, [ (\W)^n, \h ] ] = [ (\W)^{n+1}, \h ]$, and thus we obtain
\begin{align}
    \partial_s^{n+1} \left( e^{s\W} \h e^{-s\W} \right) = e^{s\W} [ (\W)^{n+1}, \h ] e^{-s\W}.
\end{align}
By induction, the formula holds for all \( n \geq 1 \). 
\end{proof}
\end{minipage}

\item Let us consider the state $\ket \psi = \h \ket {\phi} / \|\h \ket {\phi} \|$. We will use it as follows.

\begin{minipage}{0.95\textwidth}
\begin{lemma}[Second moment bound]\label{lem:secondmomentbound}
    Let $\hat A =\hat A^\dagger$, then we have $\bra\phi \h \hat A\h\ket \phi \le \bra\phi \h^2 \ket \phi \|\hat A\| $.
\end{lemma}
\begin{proof}
By definition
    \begin{align}
        \bra\phi \h \hat A\h\ket \phi = \bra\phi \h^2 \ket \phi  \bra\psi \hat A \ket \psi 
    \end{align}
    and we use the variational definition of the operator norm.
\end{proof}
\end{minipage}

\item Usually, operator norms should be used to bound the proximity of unitaries.
However, while in general $\|[\h,\Omega\|\le 2\|\h\|\,\|\Omega\|$ is the standard bound, we can get a much stronger bound as follows.

\begin{minipage}{0.96\textwidth}
\begin{lemma}[Bracket perturbation]\label{lemma: bracket perturbation} Let $\h= \h^\dagger$ and $\Omega = \ket \Omega\bra\Omega$ then $\|\1- e^{[\h,\Omega]}\|\le 2 \sqrt{V(\Omega)}$.
\end{lemma}
\begin{proof}
    Following Bhatia \cite{bhatia} or using Eq.~(115) of Ref.~\cite{double_bracket2024} one can prove
    \begin{align}
        \|\1 - e^{[\h,\Omega]}\|\le  \|{[\h,\Omega]}\|\ ,
    \end{align}
    which holds for any unitarily invariant norm, in particular the operator norm.
    By using $\|[\h,\Omega]\|\le \|[\h,\Omega]\|_\text{HS}$ and Lemma~\ref{lemma: HS norm squared of commutator} we obtain the result. 
\end{proof}
Note, that normally upper bounding the operator norm by the Hilbert-Schmidt norm gives exponentially loose bounds. This is not the case here.
\end{minipage}

 \item For QITE, we are often interested in the commutator of the state and the system Hamiltonian. In the case of pure states, we may show that its Hilbert Schmidt norm is related to thermodynamic properties of the state, in particular energy variance. 
    
    \begin{minipage}{0.96\textwidth} 
\begin{lemma}[Bracket-variance duality]\label{lemma: HS norm squared of commutator} 
For any Hamiltonian $\h= \h^\dagger$ and pure state $\Omega = \ket \Omega\bra\Omega$,
    \begin{align}
        \HSnorm{\big[\h, \Omega \big]}^2 = 2 V(\Omega) = 2 \langle \Omega | \h^2 | \Omega\rangle - 2 \Omega | \h | \Omega\rangle^2 \ .
    \end{align}
\end{lemma}
\begin{proof}
    Using the definition of Hilbert-Schmidt norm, we have
    \begin{align}
\HSnorm{[\h, \Omega]}^2 
&=  -\Tr\bigg( (\h\Omega-\Omega \h ) ( \h\Omega-\Omega \h ) \bigg)= - \Tr(\h\Omega \h\Omega -\Omega \h^2 \Omega - \h \Omega \h   +\Omega \h\Omega \h )\ .
\end{align}
We evaluate the trace in a basis that includes $\ket \Omega$ and obtain the result by collecting repeated terms.
\end{proof}
\end{minipage}

\item  The error term in the linear approximation of the QITE DBI is bounded above by the product of the step size and the energy fluctuations raised to the second power. 

\begin{minipage}{0.96\textwidth}
\begin{lemma}\label{lemma:DBI fidelity error term HS bound}
Let  $\h$ be a Hamiltonian and a density matrix $\Omega$
corresponding to a pure state.
For any real parameter $r \in \mathbb{R}$ the error associated with the linear approximation of the QITE DBR     \begin{align}\label{eq: DBI fidelity error term HS bound}
         \mathcal R_r\coloneqq e^{-r[\h, \Omega ]}-\1 + r[\h, \Omega ]
    \end{align}
satisfies $\HSnorm{\mathcal R_r} \leq r^{2} \V_k$.
\end{lemma}
\end{minipage}
\begin{minipage}{0.96\textwidth}
\begin{proof}
    Using Taylor's theorem we have
    \begin{align}
    \mathcal R_r = \mathcal R_0 + \int_0^r dx \; \partial_x \mathcal R_x  =   \int_0^r dx \; [\h, \Omega ]\left(\1- e^{-x[\h, \Omega ]}\right) \ ,
\end{align}
where we use the fact that $\mathcal R_0=0$. 
Taking the Hilbert-Schmidt norm for the last expression becomes
\begin{align}
 \HSnorm{ \mathcal R_r} =  \HSnorm{ \int_0^r dx \;[\h, \Omega ]\big(\1- e^{-x[\h, \Omega ]}\big) } 
     \leq \HSnorm{\big[\h,\Omega \big]} \times \HSnorm{\int_0^r dx \;(\1- e^{-x[\h, \Omega ]}) }   \ .
\end{align}
Set $[H, \Omega ] \to x [H, \Omega ]$ in Lemma.~\ref{lemma: bracket perturbation}, we have
\begin{align}
    \HSnorm{\1- e^{-x[H, \Omega ]} } \leq x\HSnorm{[H, \Omega ]} \ .
\end{align}

Thus, it simplifies to
\begin{align}
     \HSnorm{\mathcal R_r} &  \leq  \HSnorm{\big[H,\Omega \big]}^2 \times \HSnorm{\int_0^r dx\; x}  = \frac{r^2}{2} \HSnorm{\big[H,\Omega \big]}^2\ .
\end{align}
Using Lemma.~\ref{lemma: HS norm squared of commutator}, we then obtain $\HSnorm{\mathcal R_r}  \leq r^2 \V_k$.
\end{proof}
\end{minipage}

\end{enumerate}

\subsection{Proof of fluctuation-refrigeration relation for QITE DBI and its cooling rate}\label{app: DBI cooling rate}
In this section, we derive the cooling rate of QITE state by using DBI implementation.

\begin{proposition}[QITE-DBI Fluctuation-refrigeration relation]\label{prop: DBI frr}
    Let $\lbrace\ket{\sigma_{k}}\rbrace_{k}$ be a sequence of states that solve the QITE-DBI relation in Eq.~\eqref{eq: DBI state main}.
    The first-order cooling rate of QITE DBI equals that of QITE
\begin{align}
    \E_{k+1}-\E_k \leq -2s_k \V_k+\mathcal O(s_k^2)\ .
\end{align}
If each time step is chosen such that $\displaystyle s_k \leq \dfrac{\V_k}{4 \|\h\| \langle \sigma_k| \h^2 |\sigma_k\rangle } $, then the DBI cooling rate is lower bounded as
\begin{align}
    \E_{k+1}-\E_k \leq -s_k \V_k\ ,
\end{align}
where the average energies $\E_k$ and variances $\V_k$ are defined in Eq.~\eqref{eq:qite_dbi_e&v}.
\end{proposition}
\begin{proof}
Let us begin by defining the average energy of the $k$-th DBI state as a function of the time step $s$,
\begin{align}\label{eq:Eks_fn}
    \E_k(s) = \langle \sigma_k|e^{-s[\sigma_k,\h]}\h e^{s[\sigma_k,\h]}|\sigma_k\rangle\ .
\end{align}
Note that the derivative of $\E_k(s)$ when evaluated at $s=0$ gives $\partial_s \E_k(0) = -2\V_k$, exactly as in Proposition.~\ref{prop:frr_QITE}. 
The difference in factor 2 here is that we need to provide a lower bound to the cooling rate for a finite double-bracket rotation duration $s\neq 0$.
    
By using a Taylor expansion on Eq.~\eqref{eq:Eks_fn}, we find
\begin{align}
    \E_k(s_k) = \E_k + s_k \partial_s \E_k(s)_{|s=0} +\frac 12 s_k^2\partial_s^2 \E_k(s)_{|s=\xi_k}\ ,
\end{align}
where $\xi_k \in [0,s_k]$ is the point of the Lagrange remainder (point 6, Sec.~\ref{sec: terminology}). 
Next, recall that $\E_k(s_k) $ is simply equal to $\E_{k+1}$ by definition, hence
\begin{align}
    \E_{k+1} - \E_k = s_k \partial_s \E_k(s)_{|s=0} +\frac 12 s_k^2\partial_s^2 \E_k(s)_{|s=\xi_k}\ .
\end{align}
To simplify the notation, it is convenient to use the following shorthands
\begin{align}
    \hat{W}_k &=[\sigma_k,\h], \qquad 
    \h_k(s) =e^{-s\hat{W}_k} \h e^{s\hat{W}_k}\ .
\end{align}
Note that $e^{-s\hat{W}_k}$ is unitary, and hence $\|\h_k(s)\|=\|\h\|$. The 1st order derivative of average energy is then given by $\partial_s \E_k(s) 
= \langle\sigma_k| \big[\h_k(s),\hat{W}_k \big]|\sigma_k\rangle$. Expanding this yields for $s=0$
\begin{align}
    {\partial_s \E_k(s)}{\big|_{s=0} } &=\langle \sigma_k| \big(\h \sigma_k\h - \h ^2\sigma_k -\sigma_k\h^2 +\h \sigma_k \h\big)|\sigma_k\rangle \\
    &=-2  \langle \sigma_k|\h^2|\sigma_k\rangle + 2 \E_k^2 = -2\V_k \leq 0 \ .\label{eq: DBI 1st order is variance}
\end{align}
For the 2nd order derivative, we use a double-nested commutator Eq.~\eqref{eq: nth derivative of H evolution} and by explicit computation
\begin{align}
     \partial_s^2 \E_k(s) 
    &= \langle \sigma_k|[[\h_k(s),\hat{W}_k],\hat{W}_k] |\sigma_k\rangle \\
    &=\langle \sigma_k|\h_k(s)\hat{W}_k^2|\sigma_k\rangle + \langle \sigma_k|\hat{W}_k^2 \h_k(s)|\sigma_k\rangle -2 \langle \sigma_k|\hat{W}_k \h_k(s)\hat{W}_k |\sigma_k\rangle  \ ,
\end{align}
where we commuted $\hat{W}_k$'s with its exponentials.
Next, we will use
\begin{align}
    \hat{W}_k^2 |\sigma_k\rangle = (\sigma_k\h\sigma_k\h-\sigma_k \h^2)|\sigma_k\rangle
\end{align} to expand the DBI brackets. This gives us
\begin{align}
    \partial_s^2 \E_k(s) =& \langle \sigma_k|\h_k(s)\sigma_k \h  \sigma_k \h|\sigma_k\rangle  -  \langle \sigma_k|\h_k(s)\sigma_k \h^2 |\sigma_k\rangle 
 + \langle \sigma_k|\h  \sigma_k \h \sigma_k
 \h_k(s)|\sigma_k\rangle -  \langle \sigma_k| \h^2 \sigma_k \h_k(s) | \sigma_k \rangle \nonumber\\
 &  - 2\langle \sigma_k|\h \h_k(s) \sigma_k \h  |\sigma_k\rangle +2\langle \sigma_k| \h \h_k(s) \h|\sigma_k\rangle+  2  \langle \sigma_k| \h  \sigma_k \h_k(s)\sigma_k\h  |\sigma_k\rangle- 2\langle \sigma_k| \h  \sigma_k \h_k(s) \h|\sigma_k\rangle\\
 &= -2 \V_k \E_k(s)  + 2\langle \sigma_k| \h \h_k(s) \h|\sigma_k\rangle +2 \E_k^2 \E_k(s)-4 \E_k \text{Re}(\langle\sigma_k| \h \h_k(s)|\sigma_k \rangle )
 \ .
 \label{eq:2ndorderfluct}
\end{align}
The first term is negative so reduces the energy of the state (to 2nd order).
The second term can be bound using Lemma \ref{lem:secondmomentbound} as
\begin{align}
    |\langle \sigma_k| \h \h_k(s) \h|\sigma_k\rangle| 
    &\le \langle \sigma_k| \h^2 |\sigma_k\rangle \|\h_k(s)\| 
    \le \langle \sigma_k| \h^2 |\sigma_k\rangle \|\h\|\ .
\end{align}
To bound the last term in Eq.~\eqref{eq:2ndorderfluct}, we have
\begin{align}
    |\text{Re}(\langle\sigma_k| \h \h_k(s)|\sigma_k \rangle )| \le |\big\langle \h|\sigma_k \rangle,\h_k(s)|\sigma_k \rangle\big\rangle| &\le \sqrt{\langle \sigma_k| \h^2 |\sigma_k\rangle} \,\sqrt{\langle \sigma_k| \h_k(s)^2 |\sigma_k\rangle} \\ 
    &\le \sqrt{\langle \sigma_k| \h^2 |\sigma_k\rangle} \,\|\h\|\ ,
\end{align}
where we used Cauchy-Schwarz in the second inequality. Therefore, we have
\begin{align}\label{eqbound:2nd_derivative_Ek}
     \partial_s^2 \E_k(s) &\leq  -2 \V_k \E_k(s)  + 2\E_k^2 \|\h\| + 4\E_k \sqrt{\langle \sigma_k| \h^2 |\sigma_k\rangle} \,\|\h\| +2 \langle \sigma_k| \h^2 |\sigma_k\rangle \|\h\|
     \leq   8 \langle \sigma_k| \h^2 |\sigma_k\rangle \|\h\|.
\end{align}
As a remark, from the end of Eq.~\eqref{eqbound:2nd_derivative_Ek} we can obtain a weaker bound by using the sub-multiplicativity of the operator norm, i.e. we have $\partial_s^2 \E_k(s) 
     \leq 8\|\hat H \|^3$. Finally, we have 
\begin{align}
    \E_{k+1} - \E_k
    &\le - 2 s_k \V_k +4s_k^2 \langle \sigma_k| \h^2 |\sigma_k\rangle \|\h\| \ ,\label{DBI rate strengthening}
\end{align}
so as we place the constraint that $-s_k\V_k+4s_k^2\langle \sigma_k| \h^2 |\sigma_k\rangle \|\h\| \leq 0$, this leads to $ \E_{k+1} - \E_k \le -s_k \V_k$. Evaluating the constraint as an upper bound on $s_k$, we obtain $\displaystyle s_k \leq \;  \dfrac{\V_k}{4 \|\h\| \langle \sigma_k| \h^2 |\sigma_k\rangle } $. 
\end{proof}
\begin{remark}
    Scrutinizing Eq.~\eqref{eq:2ndorderfluct} around $s=0$ allows us to further see the significance of higher moments of energy distribution, for the purposes of QITE. In particular, using Eq.~\eqref{eq:2ndorderfluct}, we have
    \begin{align}
    \partial_s^2 \E_k(s)\big|_{s=0}
 &= -2 \V_k \E_k  + 2\langle \sigma_k| \h ^3|\sigma_k\rangle +2 \E_k^2 \E_k-4 \E_k \langle\sigma_k| \h ^2|\sigma_k\rangle \\
 &= 2\bigg(\langle  \sigma_k| \h^3 |  \sigma_k\rangle - \V_k \E_k  +\E_k^3 -2 \E_k(\V_k+\E_k^2)\bigg) \\
 &= 2\bigg( \langle  \sigma_k | \ \h^3 |  \sigma_k \rangle -3 \E_k \V_k -  \E_k^3\bigg) =2\langle\sigma_k |  (\h-\E_k)^3 |\sigma_k\rangle, \label{eq: DBI 2nd order is skewness}
\end{align}
because one can verify that
\begin{align}
    \langle  \sigma_k| (\h-\E_k)^3 |  \sigma_k \rangle &=  \langle  \sigma_k| \h^3 |  \sigma_k\rangle  - 3 \langle  \sigma_k| \h^2 \E_k  |  \sigma_k \rangle+3 \langle  \sigma_k|  \h \E_k^2 |  \sigma_k \rangle - \langle  \sigma_k|  \E_k^3 |  \sigma_k\rangle\\
    &=  \langle  \sigma_k |\ \h^3 |  \sigma_k \rangle -3 \E_k \V_k -  \E_k^3.
\end{align}
Combining  Eq.~\eqref{eq: DBI 1st order is variance} and  Eq.~\eqref{eq: DBI 2nd order is skewness}, the Taylor expansion of the expected energy is given by
\begin{align}
    \E_k(s) = \E_k -2 s \langle\sigma_k|  (\h-\E_k)^2 |\sigma_k\rangle  + s^2  \langle\sigma_k|  (\h-\E_k)^3 |\sigma_k\rangle +\mathcal{O}(s^3) \ .
\end{align}
This means that whenever skewness is negative, then one can analytically justify longer double-bracket rotation  durations -- the second order term  would then enhance cooling process for the state.
\end{remark}

\subsection{Exponential fidelity convergence of QITE DBI}

\begin{lemma}[{Lower bound to energy over variance}]\label{lemma: E> V/H}
Let $| \Omega \rangle$ be any pure state. 
Denote Hamiltonian as $ \h= \sum_{j=0}^{d-1} \lambda_j |\lambda_j \rangle \langle \lambda_j |$ where we assume that the eigenvalues of $\hat H$ are ordered increasingly and that $\lambda_0=0$.
Suppose that the ground-state fidelity is given by $F=1-\epsilon$, then we have
  \begin{align}
      \dfrac{E}{V}\geq\dfrac{1}{\|\h\|}  \ ,
  \end{align}
  where $E=\langle \Omega  |\h^2|\Omega \rangle $ is the expected energy and $V = \langle\Omega |\h^2|\Omega  \rangle - E^2$ is the energy variance.
\end{lemma}
\begin{proof}
Let us define the probability to occupy the $i$-th eigenstate $p_i \coloneqq \lvert \langle \lambda_{i}\vert\Omega\rangle\rvert^{2}$.  
Since $\lambda_0=0$ by assumption, we obtain 
    \begin{align}
        E-\dfrac{V}{\|\h\|} &= \sum_{i=1}^{d-1}\lambda_i p_i - \frac{1}{\|\h\|}\sum_{i=1}^{d-1}\lambda_i^2 p_i + \frac{1}{\|\h\|}\left(\sum_{i=1}^{d-1}\lambda_i p_i\right)^2\\
        &\geq \sum_{i=1}^{d-1}\left (1-\frac{\lambda_i}{\|\h\|}\right)\lambda_i p_i \geq 0 \ .
    \end{align}
\end{proof}

\begin{remark}
    This bound is (asymptotically) tight for $\epsilon\rightarrow 0$. Indeed, if $p_{d-1}=\epsilon$ and $p_i=0$ for $i<d-1$, we obtain 
    \begin{align}
        \dfrac{E}{V} = \dfrac{\lambda_{d-1}\epsilon}{\lambda_{d-1}^2\epsilon-\lambda_{d-1}^2\epsilon^2}=\dfrac{1}{\|\h\|(1-\epsilon)}
    \end{align}
\end{remark}

This lemma can be used to prove the following result.

\begin{theorem}[Exponential fidelity convergence for QITE DBI ]\label{th: dbi fidelity revised}
    Let $\h= \sum_{j=0}^{d-1} \lambda_j |\lambda_j \rangle \langle \lambda_j | $
    be a Hamiltonian where $d=\text{dim}(\h)$, $\{ |\lambda_i \rangle \}$ is the set of increasingly ordered eigenstates with $\lambda_0=0$ and we denote the ground state by \( |\lambda_0\rangle \). 
   Let $\ket{\sigma_0}$ be an initial state with non-zero overlap to the ground-state, where the fidelity is given by $F_0 = | \langle \lambda_0 | \sigma_0 \rangle|^2$. 
 Consider QITE DBI with 
    \begin{align}
    | \sigma_{k+1} \rangle = e^{s_0[\ket{\sigma_k}\bra{\sigma_k},\h]} | \sigma_k \rangle\ ,
    \end{align}
    where $s_0=\sqrt{F_0} / (4\lambda_{d-1})$ is the same in all steps.
    The fidelity for the $k$-th DBI step defined by $F_k = | \langle \lambda_0 | \sigma_k \rangle|^2$ satisfies
\begin{align}\label{eq: q in thm}
    F_k \ge 1- q^k\ , \text{ where} \quad q = 1- \left(\frac 14\frac{\lambda_1}{\lambda_{d-1}}F_0^{\frac 32}\right)
    =  1- \left(\frac 14\frac{\lambda_1}{\|\h\|}F_0^{\frac 32}\right)\ .
\end{align}
\end{theorem}
The convergence rate is set by a quantity related to the condition number of the Hamiltonian which also appears in classical algorithms.
For a non-singular matrix $A$ the condition number is defined as $\kappa(A) = \|A^{-1}\|\,\|A\|$. In our problem, we have affine invariance $\h \mapsto \h+\alpha \mathbb{I}$ which does not change the ground-state state vector $\ket{\lambda_0}$ or the dynamics of the Hamiltonian.
    Observe that DBIs are invariant under affine shifts of the Hamiltonian, in particular we have $ [\h- \lambda_0\1,\sigma_k] =  [\h, \sigma_k]$
which means that $\h' = \h - \lambda_0 \1$ with vanishing ground-state energy $\h' \ket{\lambda_0} = 0$ has the same DBI unitaries.
Thus, without loss of generality, we assume the ground-state to be zero, i.e. $\lambda_0=0$.
In that case, a more meaningful notion of the condition number is
\begin{align}
\kappa_0(\h) = \lambda_{d-1} / \lambda_1\ . 
\end{align}
In classical numerical analysis we have $\kappa(A)\ge 1$ and whenever $\kappa(A)$ is large then $A$ is considered badly conditioned.
Using this notation allows us to write
\begin{align}
    q = 1- \left(\frac 14 \frac {F_0^{\frac 32}}{\kappa_0(\h)} \right)\ .
\end{align}
If $\kappa_0(\h)$ is large, which for local Hamiltonians corresponds to a small spectral gap, then the QITE DBI converges more slowly.

\begin{proof}

Next, the ground-state fidelity after each round of QITE DBI iteration is given by 
\begin{align}
   F_{k+1} =  \lvert  \langle \lambda_0 | \sigma_{k+1}\rangle\rvert^{2} 
   &= \bigg\lvert \langle \lambda_0 | \sigma_{k} \rangle - s_0\bra{\lambda_0}[\h, \sigma_k ] |\sigma_k \rangle + \bra{\lambda_0}R|\sigma_k \rangle \bigg \rvert^{2}\\
    &= \bigg \lvert \langle \lambda_0 | \sigma_{k} \rangle(1+s_0\langle \sigma_k | \h | \sigma_k\rangle) -s_0 \langle \lambda_0 | \h | \sigma_k \rangle + \bra{\lambda_0}R|\sigma_k  \rangle \bigg \rvert^{2}\\
    &= \bigg \lvert \langle \lambda_0 | \sigma_{k} \rangle(1+s_0\E_k) + \bra{\lambda_0}R|\sigma_k  \rangle \bigg \rvert^{2},
\end{align}
where $R$ denotes the error term of the linear approximation
to the QITE DBI. 
Since we set $\lambda_0=0$, the term $s_0 \langle \lambda_0 | \h | \sigma_k \rangle$ vanishes.
Furthermore, note that $|\langle \lambda_0 | \sigma_k \rangle|^2 =F_k $. 
Hence we have 
\begin{align} \label{eq:fidelity_to_ground}
      F_{k+1} &=F_k(1+s_0\E_k)^{2} + 2(1+s_0 \E_k)\mathrm{Re} \bigg(\langle \lambda_0 | \sigma_{k} \rangle  \langle \sigma_k |R^\dagger |\lambda_0 \rangle \bigg) + \lvert \bra{\lambda_0}R|\sigma_k\rangle\rvert^{2} \ .
\end{align}
Next, note that $0 \leq \lvert \bra{\lambda_0}R|\sigma_k\rangle\rvert^{2} \leq  \lVert R \rVert_{\rm HS} = \lVert R^\dagger \rVert_{\rm HS} \le s_0^2 V_k$ by Lemma.~\ref{lemma:DBI fidelity error term HS bound}.
Thus,  Eq.~\eqref{eq:fidelity_to_ground} becomes 
\begin{align}\label{eq:fidelity_to_ground_2}
    F_{k+1} 
    &\geq F_k (1+s_0\E_k)^{2} - 2s_0^2 \V_k\bigg(1+s_0\E_k \bigg)\sqrt{F_k}  \ ,
\end{align}
where we dropped the third term in Eq.~\eqref{eq:fidelity_to_ground}, and use the bound $\langle \sigma_k |R^\dagger |\lambda_0 \rangle \geq - \lVert R^\dagger \rVert_{\rm HS} \ge -s_0^2 \V_k$. 

Using the definition of the infidelity $F_k =1 - \epsilon_{k}$, we can rewrite it as 
\begin{align}
1-\epsilon_{k+1} &\geq (1-\epsilon_{k}) (1+s_0\E_k)^{2} -2 s^2_0 \V_k \bigg(1+s_0\E_k\bigg)\sqrt{1-\epsilon_{k}} \\
 \epsilon_{k+1} &\leq 1 -(1-\epsilon_{k}) (1+s_0\E_k)^{2} +2 s^2_0 \V_k \bigg(1+s_0\E_k \bigg)\sqrt{1-\epsilon_{k}}\\
  &\leq \epsilon_k-2s_0 \E_k (1-\epsilon_k)-s_0^2\E_k^{2}(1-\epsilon_k)  +2 s^2_0 \V_k \bigg(1+s_0\E_k \bigg)\sqrt{1-\epsilon_{k}}\\ 
  &\leq \epsilon_k-2s_0\E_k (1-\epsilon_k) +2 s^2_0 \V_k \bigg(1+s_0\E_k \bigg)\sqrt{1-\epsilon_{k}}\ . \label{eq: DBI fidelity before approximation}
\end{align}
Next, note that according to the assumed form $s_0$ in the theorem statement, and the fact that $\E_k \leq \|H\| = \lambda_{d-1}$, we have that $s_0\E_k \le \dfrac 14 \dfrac{\sqrt{F_0}}{\lambda_{d-1}} \lambda_{d-1}  \le 1$, so we can simplify the above bound by relaxing it further:
\begin{align} 
 \epsilon_{k+1} 
  \leq \epsilon_k-2s_0\E_k (1-\epsilon_k) +4 s^2_0 \V_k \sqrt{1-\epsilon_{k}}\ .
   \label{eq: DBI fidelity improvement}
\end{align}
Our goal is to suppress the second-order term by utilizing half of the first-order infidelity reduction. To do so, the condition $4 s^2_0 \V_k\sqrt{1-\epsilon_k}  \le  s_0\E_k (1-\epsilon_k)$
  must be satisfied, which is equivalent to $s_0  \le  \E_k \sqrt{1-\epsilon_k} /(4 \V_k)$. 
  Next, from Lemma.~\ref{lemma: E> V/H}, we have $\E_k/\V_k \ge 1/\lambda_{d-1}$.
Using this bound, we can verify that indeed,
\begin{align}
    s_0 = \frac 14\frac{\sqrt{F_0}}{\lambda_{d-1}} \le \frac{\E_k}{4\V_k} \sqrt{1-\epsilon_k}\ .
    \label{eq: DBI fidelity s0 bound}
\end{align}
Since the condition  $s_0  \le  \E_k \sqrt{1-\epsilon_k} /(4 \V_k)$ holds, Eq.~\eqref{eq: DBI fidelity improvement} can be directly upper-bounded by 
\begin{align}
    \epsilon_{k+1} \leq \epsilon_{k} - s_0\E_{k}F_0\ . 
    \label{eq: DBI fidelity improvement II}
\end{align}
In Eq.~\eqref{eq: DBI fidelity s0 bound} we assume that $\epsilon_k \le \epsilon_0$. Clearly, this is satisfied for $k=0$, and can be concluded for $k>0$ by induction on $k$ using Eq.~\eqref{eq: DBI fidelity improvement II}.
Finally, using the definition of $s_0$, we have that the error reduces to
\begin{align}
    \epsilon_{k+1} &\leq \epsilon_{k} -  \frac 14\frac{\lambda_1}{\lambda_{d-1}}F_0^{\frac 32}\epsilon_{k}  = \left(1 - \left(\frac 14\frac{\lambda_1}{\lambda_{d-1}}F_0^{\frac 32}\right)\right)\epsilon_{k}\ , 
\end{align}
where we use the relation $\E_{k} \geq \lambda_1 \epsilon_k$ from Lemma.~\ref{lemma: Ek > lambda epsilon}.
\end{proof}

\section{Double-Bracket Quantum Imaginary-Time Evolution (DB-QITE)}\label{sec: DBQA}

In Sec.~\ref{sec: QITE} and Sec.~\ref{sec: QITE DBI}, we consider the QITE realization via the continuous DBF and discrete DBI method.
Here, we aim to employ the discrete group commutator iteration (GCI) to approximate DBI.
As demonstrated in Lemma.~(9) from \cite{double_bracket2024}, we can approximate DBI steps by group-commutators with an error term of magnitude $\mathcal{O}(s_k^{3/2})$, i.e. we have
\begin{align}\label{eq:GCI_compilation}
  \left\|  e^{i\sqrt{s_k}\h}e^{i\sqrt{s_k}\omega_k}    
    e^{-i\sqrt{s_k}\h}
    e^{-i\sqrt{s_k}\omega_k}- e^{s_k[\omega_k,\h]} \right\| \leq   s_k^{3/2} \bigg( \|[\h, [\h, \omega_k]]\| + \|[\omega_k, [\omega_k, \h]]  \|\bigg) \ ,
\end{align}
where $\|.\|$ represents the operator norm, and $k$ denotes the number of iteration in both DBI and GCI. Next, observe that $e^{-i\sqrt{s_k}\omega_k} \ket{\omega_k} = e^{-i\sqrt{s_k}} \ket{\omega_k} $ which motivates the definition of the \emph{reduced group-commutator iteration (DB-QITE)}: 
\begin{align}\label{eq:reducedGCI}
    \ket{\omega_{k+1}}
    &=   e^{i\sqrt{s_k}\h}e^{i\sqrt{s_k}\omega_k}    
    e^{-i\sqrt{s_k}\h}
    \ket{\omega_k} \ .
\end{align}
This is essentially Eq.~\eqref{eq: GCI state main} in the main text, but in general one can optimize the individual step sizes $s_k$ in each round. 
Lastly, we denote the expected energy and variance after the $k$-th round of DB-QITE as
\begin{align}\label{eq: mean energy and variance for DB-QITE}
    E_k = \langle \omega_k|\h|\omega_k\rangle \qquad \text{and} \qquad  V_k=  \langle \omega_k |\h^2|\omega_k \rangle - \langle \omega_k|\h^2|\omega_k\rangle^2 \ .
\end{align}
Intuitively due to Eq.~\eqref{eq:GCI_compilation}, one expects a qualitatively similar behavior between DB-QITE and DBF/DBI. We derive quantitative bounds in the next sections, to compare their explicit differences. 

Before proving the main results in this section, i.e. analytical guarantees for energy loss and fidelity improvement in DB-QITE implementation, let's first present two key lemmas that will be handy in the proof.

\begin{lemma}
    Suppose that the DB-QITE state $\ket{\omega_k}$ is a pure state, then we have the following equivalent representation for the 
 next DB-QITE state $\ket{\omega_{k+1}}$:
    \begin{align}
        \ket{\omega_{k+1}}
    &=   e^{i\sqrt{s_k}\h}e^{i\sqrt{s_k}\omega_k}    
    e^{-i\sqrt{s_k}\h}
    \ket{\omega_k} 
    \iff
        \ket{\omega_{k+1}} = \left(\1 - (\1-e^{i\sqrt{s_k}}) \phi(-\sqrt{s_k}) e^{i\sqrt{s_k}\h}\right) \ket{\omega_{k}} \ , \label{eq: more compact form of DBQA states}
    \end{align}
    where we define the characteristic function as
    \begin{align}\label{eq: DBQA characteristic function}
        \phi(t)\coloneqq \bra{\omega_{k}}e^{it\h}\ket{\omega_{k}} \ .
    \end{align}
\end{lemma}
\begin{proof}
Since we consider a pure state $\omega_{k}=\dm{\omega_{k}}$, we obtain the following identity
\begin{align}
     e^{i\sqrt{s_k}\omega_{k}} = \1 - (1-e^{i\sqrt{s_k}})\omega_{k} \ .
\end{align}
Therefore, the DB-QITE recursion can be simplified to
\begin{align}
    \ket{\omega_{k+1}} &=  e^{i\sqrt{s_k}\h} \left( \1 - (1-e^{i\sqrt{s_k}})\omega_{k}\right)  
    e^{-i\sqrt{s_k}\h}
    \ket{\omega_k} \\
    &= \ket{\omega_{k}}-(1-e^{i\sqrt{s_k}}) e^{i\sqrt{s_k}\h} \dm{\omega_{k}} e^{-i\sqrt{s_k}\h}
    \ket{\omega_k}  \\
    &=\left(\1 - (1-e^{i\sqrt{s_k}}) \phi(-\sqrt{s_k}) e^{i\sqrt{s_k}\h}\right) \ket{\omega_{k}} \ , 
\end{align}
where we use the density matrix representation in the second line and the definition of the characteristic function in the last line.
\end{proof}

\begin{lemma}\label{lemma: characteristic function bound}
Denote the characteristic function as $   \phi(t)\coloneqq \bra{\omega_{k}}e^{it\h}\ket{\omega_{k}} $.
Then the $n$-th derivative of $\phi(t)$ can be upper-bounded by
\begin{align}\label{eq: phi simpler bound}
    \lvert \phi^{(n)}(\xi) \rvert  \leq \|\h^{n}\| \ ,
\end{align}
where $\xi \in [0,t]$.
Moreover, suppose one knows the ground-state infidelity $\epsilon_k = 1- F_k$ at $k$-th QITE DBQA iteration, then we can obtain a tighter bound for $\lvert \phi^{(n)}(\xi) \rvert $, i.e.
\begin{align}\label{eq: phi tighter bound}
      \lvert \phi^{(n)}(\xi) \rvert  \leq \epsilon_k \|\h^{n}\| \ .
\end{align}
\end{lemma}

\begin{proof}
Directly evaluating the $n$-th order derivative of $\phi(t)$ gives
\begin{align}
     \phi^{(n)}(\xi)= i^n \bra{\omega_{k}} e^{i\xi\h}\h^{n}\ket{\omega_{k}} \ .
\end{align}
As the operator norm is equal to the largest eigenvalue, we obtain the bound
\begin{align}
		\lvert \phi^{(n)}(\xi) \rvert = \lvert \bra{\omega_{k}} e^{i\xi\h}\h^{n}\ket{\omega_{k}}\rvert  \leq \|e^{i\xi\h}\h^{n}\| = \|\h^{n}\| \ ,
        \end{align}
        where we use the unitary invariant property of the operator norm in the last equality and neglect the factor $i^n$ as it is of norm 1.
        Thus, the first part of this lemma has been proven.
        
    Next, to prove the second part, we denote $\Pi_0, \Pi_\perp$ as the ground-state projector and its complement (i.e. $\Pi_0+ \Pi_\perp=\1$), then $F_k = \bra{\omega_k}\Pi_0\ket{\omega_k}$ and $\epsilon_k = \bra{\omega_k}\Pi_{\perp}\ket{\omega_k}$. 
    Therefore, we obtain
\begin{align}
		\lvert \phi^{(n)}(\xi) \rvert = \lvert \bra{\omega_{k}} e^{i\xi\h}\h^{n}\ket{\omega_{k}}\rvert &\leq |\bra{\omega_k}\Pi_0e^{i\xi\h}H^n\ket{\omega_k}|+|\bra{\omega_k}\Pi_{\perp}e^{i\xi\h}H^n\ket{\omega_k}|\\
        &= |\bra{\omega_k}\Pi_{\perp}e^{i\xi\h}H^n\ket{\omega_k} | \leq \epsilon_k \|\h^n\|\ , 
        \end{align}    
        where in the last inequality, one may expand the projector and use the fact that $\h^n e^{i\xi\h}\ket{\lambda_j} = \lambda_j^n e^{i\xi\lambda_j}\ket{\lambda_j}$ and that $\lambda_{j}^{n}\leq\lambda_{d-1}^{n}$ (recall that $\lambda_{d-1}$ is the largest energy eigenvalue).
\end{proof}

\subsection{Fluctuation-refrigeration relation of DB-QITE}\label{app DB-QITE fluc ref}
\newcommand{\avgE}[1]{E_{#1}}
\newcommand{\varE}[1]{V_{#1}}
\begin{theorem}\label{thm:QITE_DBQA_FRR}
	DB-QITE satisfies the fluctuation-refrigeration up to correction terms,
    \begin{align}\label{qite_dbqa_cooling_thm}
        E_{k+1} \le  E_k - 2s_k V_k + \mathcal O (s_k^{2})\ .
    \end{align}
    In particular, if the step sizes are chosen such that 
	$s_k  \le 2V_{k}\cdot \left[5\epsilon_{k}\lVert \h\rVert^{4}\right]^{-1}$,
	then $\avgE{k+1}\le \avgE{k} - s_k\varE{k}$.
\end{theorem}
\begin{proof}
	We start from RHS of Eq.~\eqref{eq: more compact form of DBQA states}, i.e.
    \begin{align}
         \ket{\omega_{k+1}} = \left(\1 - (1-e^{i\sqrt{s_k}}) \phi(-\sqrt{s_k}) e^{i\sqrt{s_k}\h}\right) \ket{\omega_{k}} \ ,
    \end{align}
    where we employ the same notation for the characteristic function defined in Eq.~\eqref{eq: DBQA characteristic function}, i.e.
    \begin{align}\label{eq:char_fn}
        \phi(t)\coloneqq \bra{\omega_{k}}e^{it\h}\ket{\omega_{k}}\ .
    \end{align}
    To simplify calculational notation, we drop the square root and index from step sizes, i.e. $\sqrt{s_k}\rightarrow t$ for the moment. Moreover, we define $c = (1-e^{it}) \phi(-t)$ and hence  we have
    \begin{align}
        \ket{\omega_{k+1}} = (\1-ce^{it\h})\ket{\omega_{k}} \ .
    \end{align}

	To derive the cooling rate, we calculate
	\begin{align}
		\avgE{k+1} = \bra{\omega_{k+1}}\h\ket{\omega_{k+1}} &= \bra{\omega_k}\left[\1-c^*e^{-it\h}\right]\cdot\h\cdot\left[\1-ce^{it\h}\right]\ket{\omega_k}\\
        & = E_k + |c|^2 E_k - 2{\rm Re} \left( \bra{\omega_k} ce^{itH} \h\ket{\omega_k}\right),\label{eq:cooling_DBQA}
	\end{align}
    where we have already made use of the fact that $\h$ commutes with $e^{-it\h}$. To achieve Eq.~\eqref{qite_dbqa_cooling_thm}, we need to derive upper bounds from Eq.~\eqref{eq:cooling_DBQA}, which we do for the individual terms:
    \begin{enumerate}
        \item The second term on the R.H.S. of Eq.~\eqref{eq:cooling_DBQA} can be upper bounded by the fact that 
    \begin{align}\label{eq:firstbound}
        |c|^2 = |(1-e^{it})\phi(-t)|^2 \leq |(1-e^{it})|^2\cdot |\phi(-t)|^2 \leq t^2,
    \end{align}
    since $(1-e^{it})(1-e^{-it}) = 2(1-\cos(t))\leq t^2$, and $|\phi(-t)|\leq \|e^{-it\h}\| \leq 1$. 
    \item The third term in Eq.~\eqref{eq:cooling_DBQA} can be rewritten as
    \begin{align}
		f(t) \coloneqq -2{\rm Re} \left( c \bra{\omega_k} e^{it\h} \h\ket{\omega_k}\right) &= -2{\rm Re}  \left( (1-e^{it})\cdot \bra{\omega_k} e^{-it\h} \ket{\omega_k} \cdot \bra{\omega_k} e^{it\h} \h\ket{\omega_k}\right) \\&=-2\mathrm{Im}\left[(1-e^{it})\phi(-t)\phi^{(1)}(t)\right]\ , 
	\end{align}
since the first derivative of $\phi(t)$ with respect to $t$ is given by $\phi^{(1)}(t) = i\bra{\omega_k} e^{it\h} \h\ket{\omega_k}$. Our grand goal is to upper bound this term. At this point, let us note that $f(t)$ is an even function with $f(0)=0$. We may then omit odd derivatives of it, and write the following Taylor expansion,
	\begin{align}
		f(t) = \frac{f^{(2)}(0)}{2}t^{2} + \frac{f^{(4)}(\xi)}{24}t^{4}. 
	\end{align}
To access the higher derivatives of $f(t)$, we start by defining 
\begin{align}\label{eq:h_fn}
    h(t) \coloneqq \phi(-t)\phi^{(1)}(t),
\end{align}
and write derivatives of $f(t)$ as
	\begin{align}
		f^{(2)}(t) &= -2\mathrm{Im}\left[e^{it}h(t) - 2ie^{it}h^{(1)}(t) + (1-e^{it})h^{(2)}(t)\right]\ ,\label{eq:f2_eval}\\
		f^{(4)}(t) &= -2\mathrm{Im}\left[-e^{it}h(t) + 4ie^{it}h^{(1)}(t) + 6e^{it}h^{(2)}(t)  - 4ie^{it}h^{(3)}(t) + (1-e^{it})h^{(4)}(t)\right]\ .\label{eq:f4_eval}
	\end{align}

    {\underline{i) Evaluating $f^{(2)}(0)$ in Eq.~\eqref{eq:f2_eval}}}\\[2pt]
	The first term is easy; we have explicitly derived $f^{(2)}(t)$ in Eq.~\eqref{eq:f2_eval}, and we need the expression for $h(t)$ as detailed in Eq.~\eqref{eq:h_fn} and \eqref{eq:char_fn}. In particular, one can verify that $h(0) = iE_k$ and $h^{(1)}(0) = -V_k$. Therefore, we have 
    \begin{align}\label{eq:f2}
        f^{(2)}(0) = -2{\rm Im} \left[h(0)-2ih^{(1)}(0)\right] = -2E_k-4V_k.
    \end{align}
    At this point, it is good to note that by combining Eqns.~\eqref{eq:firstbound} and Eq.~\eqref{eq:f2} into the equation for cooling, i.e.  Eq.~\eqref{eq:cooling_DBQA}, 
    \begin{align}
        E_{k+1}\leq E_k -2 V_kt^2 + \frac{f^{(4)}(\xi)}{24}t^{4} = E_k - V_kt^2 - V_kt^2 + \frac{f^{(4)}(\xi)}{24}t^{4} ,
    \end{align}
    which will give us the desired fluctuation-refrigeration relation $E_{k+1}\leq E_k - V_kt^2$ if we can choose $t$ such that 
    \begin{equation}\label{eq:frr_cond}
        - V_kt^2 + \frac{f^{(4)}(\xi)}{24}t^{4} \leq 0 \qquad \implies \qquad t^2 \leq \frac{24V_k}{f^{(4)}(\xi)}.
    \end{equation}
    In other words, our next step is to formulate an upper bound $f^{(4)}(\xi)\leq X $, such that by choosing $t^2 \leq \frac{24V_k}{X}$, we complete the proof.\\
    
    {\underline{ii) Upper bounding $f^{(4)}(t)$ in Eq.~\eqref{eq:f4_eval}}}\\[2pt]
    We proceed to carefully bound each individual term contained in $f^{(4)}(t)$, since
    \begin{align}
        f^{(4)}(t) &= -2\mathrm{Im}\left[-e^{it}h(t) + 4ie^{it}h^{(1)}(t) + 6e^{it}h^{(2)}(t)  - 4ie^{it}h^{(3)}(t) + (1-e^{it})h^{(4)}(t)\right] \\
        & \leq 2 \left[ ~\abso{h(t)}+ 4\abso{h^{(1)}(t)} + 6\abso{h^{(2)}(t)}+4\abso{h^{(3)}(t)} + \abso{(1-e^{-it}) h^{(4)}(t)}~ \right],\label{eq:f4bound}
    \end{align}
    where note that most $e^{it}$ terms are omitted since its norm is bounded by 1.
    Recall that $h(t)$ is defined in Eq.~\eqref{eq:h_fn} in terms of derivatives of $\phi(t)$, and so $h^{(n)}(t)$ needs to be explicitly derived as functions containing derivatives of $\phi(t)$ via chain rule.
    To bound the term $h^{(n)}(t)$, we first recall the bound of $ \lvert \phi^{(n)}(\xi) \rvert $  from Eq.~\eqref{eq: phi simpler bound} in Lemma.~\ref{lemma: characteristic function bound}, i.e. 
    \begin{align}
    \lvert \phi^{(n)}(\xi) \rvert  \leq \epsilon_{k}\|\h^{n}\| \ ,
    \end{align}
    where $\xi \in [0,t]$.

        Furthermore, let us assume that $\lVert \h\rVert\geq 1$ so that $\lVert \h^n\rVert\geq \lVert \h^{n-1}\rVert$. We summarize the bounds below:
	\begin{align}
		\lvert h(t) \rvert &\leq \avgE{k} \leq \epsilon_{k}\|\h\| \leq \epsilon_{k}\|\h^4\| \ ,\label{eq:firstterm}\\
		\lvert h^{(1)}(t) \rvert &\leq \varE{k} \leq \epsilon_{k}\|\h^2\| \leq \epsilon_{k}\|\h^4\| \ ,\\
		\lvert h^{(2)}(t) \rvert &\leq \lvert \bra{\omega_{k}}\h^{3}\ket{\omega_{k}} - \avgE{k}\bra{\omega_{k}}\h^{2}\ket{\omega_{k}}\rvert \leq  \epsilon_{k}\|\h^3\|\leq  \epsilon_{k}\|\h^4\|\ ,\\
		\lvert h^{(3)}(t) \rvert &\leq \lvert \bra{\omega_{k}}H^{4}\ket{\omega_{k}} - \bra{\omega_{k}}H^{3}\ket{\omega_{k}}\avgE{k} - 3\bra{\omega_{k}}H^{3}\ket{\omega_{k}}\avgE{k} + 3\bra{\omega_{k}}H^{2}\ket{\omega_{k}}^{2}\rvert \leq 4 \epsilon_{k}\|\h^4\|\ .
        \end{align}
        The final term requires an additional assumption that $t\lVert \h\rVert\leq 1$. With this,
        \begin{align}
		\lvert (1-e^{it})h^{(4)}(t)\rvert &\leq t\lvert -3\bra{\omega_{k}}H^{4}\ket{\omega_{k}}\avgE{k} +2\bra{\omega_{k}}H^{3}\ket{\omega_{k}}\bra{\omega_{k}}H^{2}\ket{\omega_{k}}  +\bra{\omega_{k}}H^{5}\ket{\omega_{k}}  \rvert \\
		&\leq t\bra{\omega_{k}}H^{5}\ket{\omega_{k}} +2t\max\{\bra{\omega_{k}}H^{4}\ket{\omega_{k}}\avgE{k},\bra{\omega_{k}}H^{3}\ket{\omega_{k}}\bra{\omega_{k}}H^{2}\ket{\omega_{k}}  \}\  \\
        &\leq 3t \epsilon_{k}  \|\h^5\| \leq 3
        \epsilon_{k} \|\h^4\|.\label{eq:lastterm}
	\end{align}
	Putting Eqn.~\eqref{eq:firstterm}-\eqref{eq:lastterm} back into Eq.~\eqref{eq:f4bound}, we obtain 
	\begin{align}\label{eq:boundingf4}
		f^{(4)}(\xi)\leq 60\epsilon_{k}\lVert \h\rVert^{4}\ .
	\end{align}   
    Plugging Eq.~\eqref{eq:boundingf4} into the choice of $t$ as detailed after Eq.~\eqref{eq:frr_cond} concludes the proof of the theorem (recall that $t= \sqrt{s_k}$).
    \end{enumerate}
	
\end{proof}

\subsection{Exponential fidelity convergence of DB-QITE: Proof of Theorem \ref{th: fidelity convergence}}\label{app:subsec:convergence}

Before we begin, recall that $| \lambda_0 \rangle$ denotes the ground-state.
In this section, we set the reference point $\lambda_0 = 0$, i.e. the ground-state is of zero energy.
Consequently, the spectral gap is given by $\Delta=\lambda_1-\lambda_0 =\lambda_1$.
With this assumption, we now present the proof of Theorem.~\ref{th: fidelity convergence} of main text.
\begin{theorem} \textbf{ground-state fidelity increase guarantee.}
Suppose that DB-QITE is initialized with some non-zero initial ground-state overlap $F_0$.
Let $\h$ be a Hamiltonian with a unique ground-state $| \lambda_0 \rangle$, $\lambda_0=0$, spectral gap $\Delta$ and spectral radius $\|\h\|\ge 1$.
Let $U_0$ be an arbitrary unitary and set 
  \begin{align}
      s = \frac{\Delta}{12 \|\h\|^3}\ .
  \end{align}
The states $\ket{\omega_k}:= U_k \ket 0$, where $U_k$ is defined in main text, i.e.
    \begin{align}
    U_{k+1} = e^{i\sqrt{s_k}\h}U_k e^{i\sqrt{s_k}\ket 0\langle0|}    U_k^\dagger
  e^{-i\sqrt{s_k}\h} U_k\ \ ,
  \label{DB-QITE Uk app}
\end{align}
has the ground-state fidelity lower-bounded by
\begin{align} 
    F_k := |\langle \lambda_0|\omega_k\rangle|^2  \ge  1- q^{k}\ 
\end{align}
where $q = 1- s  F_0 \Delta$.
\label{thmSimpleapp}
\end{theorem}

\begin{proof}
From RHS of Eq.~\eqref{eq: more compact form of DBQA states}, the DB-QITE recursion is given by
    \begin{align}\label{eq: GCI simplified}
         \ket{\omega_{k+1}} = \left(\1 - (1-e^{i\sqrt{s_k}}) \phi(-\sqrt{s_k}) e^{i\sqrt{s_k}\h}\right) \ket{\omega_{k}} \ ,
    \end{align}
     where we use the same notation again for the characteristic function (similar to Theorem.~\ref{thm:QITE_DBQA_FRR}), i.e.
    \begin{align}
        \phi(t)\coloneqq \bra{\omega_{k}}e^{it\h}\ket{\omega_{k}}\ .
    \end{align}
For all subsequent calculations in this proof, we define $t=\sqrt{s_k}$ to simplify the notation.
Our ultimate goal is to show that the fidelity between the ground-state and the $k$-th DBQA state ($F_k$) can be lower-bounded by
\begin{align}
    F_k \geq 1- q^{k} \ ,
\end{align}
where $q$ is a real parameter such that $0 <q < 1$.
To do so, we define the ground-state infidelity $\epsilon_{k}\coloneqq 1-F_{k}$ and we will derive a recursive inequality relating $\epsilon_{k}$ and $\epsilon_{k+1}$.

First, let us define the overlap between the ground-state and the $k$-th DBQA state as $\braket{\lambda_{0}}{\omega_{k}}$, 
which is related to the fidelity $F_k$ by $F_k= |\braket{\lambda_{0}}{\omega_{k}}|^2$.
Using Eq.~\eqref{eq: GCI simplified}, the overlap at $(k+1)$-th and $k$-th DBQA recursion is related by
\begin{align}
    \braket{\lambda_{0}}{\omega_{k+1}} &=\braket{\lambda_{0}}{\omega_{k}} - (1-e^{it}) \phi(-t)  \langle \lambda_0 | e^{it\h} | \omega_k \rangle  \\
    &= \left(1 - (1-e^{it})\phi(-t) \right) \braket{\lambda_{0}} {\omega_{k}} \ ,
    \label{eq: DBQA overlap}
\end{align}
where we use the expression $ \langle \lambda_0 | e^{it\h} = \langle \lambda_0 |e^{it \lambda_0} $ and the assumption $\lambda_0=0$ in the last line.
Using Eq.~\eqref{eq: DBQA overlap}, the  ground-state fidelities $F_{k+1} $ and $F_k$ are related by
\begin{align}
    F_{k+1} = | \braket{\lambda_{0}}{\omega_{k+1}}|^2
    = \left|g(t) \right|^2 F_k \ ,
\end{align}
where we define $ g(t)=1 - (1-e^{it})\phi(-t)$.
Using the definition of the infidelities, the ground-state infidelities $\epsilon_{k}$   and $\epsilon_{k+1}$ are related by
\begin{align}\label{eq: epsilon and p(t) relation}
    \epsilon_{k+1} = 1- |g(t) |^2 (1-\epsilon_k) = \epsilon_k-p(t)(1-\epsilon_{k})\ .
\end{align}
where we also define $p(t)=\lvert g(t)\rvert^{2}-1$ for simplicity.
Now, the remaining task is to show that $p(t) \geq c \epsilon_k$ for some positive constant $c$. 
We start by applying Taylor's theorem (see Sec.~\ref{sec: terminology}), i.e.
	\begin{align}\label{eq:f_in_Taylor}
	    p(t)= p(0)+tp^{(1)}(0)+\frac{t^{2}}{2}p^{(2)}(0) + \frac{t^{3}}{6}p^{(3)}(0) + \frac{t^{4}}{24}p^{(4)}(\xi)\ ,
	\end{align}where $\xi\in[0,t]$.
Notice that $p(t)$ is an even function with $p(0)=0$ as 
\begin{align}
  p(-t) = \lvert g(-t)\rvert^2 - 1 =\lvert g(t)^*\rvert^2 - 1 =\lvert g(t)\rvert^2 - 1= p(t) \ ,
\end{align}
 Hence, all odd order derivatives of $p(t)$ vanish, i.e. $p^{(2n+1)}(0) = 0$ for any non-negative integer $n$.
Therefore, the Taylor series  reduces to
\begin{align}\label{eq: simpler taylor series k}
	    p(t)= \frac{t^{2}}{2}p^{(2)}(0)  + \frac{t^{4}}{24}p^{(4)}(\xi) \ .
\end{align} 
\begin{enumerate}
    \item
    Directly evaluating $p^{(2)}(t)$ yields
\begin{align}
    p^{(2)}(t) &=2\mathrm{Re}[g^{*}(t)g^{(2)}(t)]+2\lvert g^{(1)}(t)\rvert^{2} \ ,
\end{align}
    To determine the expression of the term $p^{(2)}(0)$, we explicitly compute the derivatives of $g(t)$ up to second order.
    The results are given by
    \begin{align}
        g^{(1)}(t) &= ie^{it}\phi(-t) + (1-e^{it})\phi^{(1)}(-t)\ ,\label{app eq g1}\\
	g^{(2)}(t) &= -e^{it}\phi(-t) - 2ie^{it}\phi^{(1)}(-t) - (1-e^{it})\phi^{(2)}(-t)\ .\label{app eq g2}
    \end{align}
    Thus, we obtain
    \begin{align}
	p^{(2)}(0) &= 2\mathrm{Re}[g^{*}(0)g^{(2)}(0)]+2\lvert g^{(1)}(0)\rvert^{2} \\
    & = 2\mathrm{Re}[ 1 \cdot (-1 +2 E_k) ]+2\lvert i\rvert^{2} \quad = 4E_{k}\ . \label{eq: k^2 term}
\end{align}
where we use the relation $g^{(1)}(0)=i$, $\phi(0)=1$ and $\phi^{(1)}(0)= i E_k$.
    \item Here, we will derive an lower bound for 
    the term $p^{(4)}(\xi)$.

    First, recall that $p(t)=g(t) g^*(t) -1$ and hence its $n$-th order derivative is given by
        \begin{align}
        p^{(n)}(t) &=  \sum_{\tiny \makecell{n_a,n_b=1,\\n_a+n_b=n}}^n   g^{(n_a)} (t)g^{*(n_b)}(t) \ ,
    \end{align}
    where the constraint $n_a+n_b=n$ arises from the Leibniz rule.
    We can collect identical terms and rewrite it as 
    \begin{align}\label{eq: k^n combinatorics factors}
        p^{(n)}(t)
        &=  \sum_{r=0}^n \binom{n}{r}g^{(r)}(t) g^{*(n-r)}(t) \ ,
    \end{align}
    where we introduce the factor $\displaystyle  \binom{n}{r} $ to account for combinatorial degeneracy.
    In particular, for $n=4$, we have
    \begin{align}\label{eq: k^4}
         p^{(4)}(t)= g(t) g^{*(4)}(t) + 4 g^{(1)}(t) g^{*(3)}(t)+ 6 g^{(2)}(t) g^{*(2)}(t) + 4 g^{(3)}(t) g^{*(1)}(t)+ g^{(4)}(t) g^*(t) \ .
    \end{align}
    To obtain the lower bound for $  p^{(4)}(\xi)$, we first derive an upper bound for $ |g^{(r)}(\xi)|$ with arbitrary non-negative integer $r$.

      \begin{enumerate}[(a)]
          \item  In this part, we will determine the upper bound for $|g(\xi)|$. 
          
          Recall that $ g(t)=1 - (1-e^{it})\phi(-t)$ and we define $a(t)=- (1-e^{it})$ for notational simplicity.
          For the factor $a(\xi)$, it gives
          \begin{align}
              |a(\xi)|= |1-e^{i\xi}| = \sqrt{(1-\cos \xi)^2 + \sin^2 \xi} 
              =\sqrt{2-2\cos \xi}
              =\sqrt{4  \sin^2 \frac{\xi}{2}}
              \leq 2 \left\lvert \sin \frac{\xi}{2}\right\rvert 
              \leq \xi \ , \label{eq: 1-exp bound}
          \end{align}
          where we use the trigonometric identities $\cos \xi = 1- 2\sin^2 \frac{\xi}{2} $ in the last equality and we use the relation $| \sin \frac{\xi}{2}| \leq \frac{\xi}{2}$ in the last inequality. 
          Thus, we have
          \begin{align}
              |g(\xi)| \leq 1+ |a(\xi)| \lvert \phi(-\xi)\rvert \leq 1 + \xi 
        \leq 2 \ ,  \label{eq: g(t) <2}
          \end{align}
          where we use the triangle inequality in the first inequality and we use the bound $\lvert \phi(-\xi)\rvert \leq 1$. 
          Note that $\xi \leq t \leq 1 $ by assumption.

          \item  Next, we compute the upper bound for $| g^{(r)}(\xi)|$.
          For $r$-th order derivatives of $g(t)$, we have
      \begin{align}
       | g^{(r)}(\xi) | = \left\lvert \sum_{m=0}^r \binom{r}{m}a^{(m)}(\xi) \phi^{(r-m)}(-\xi) \right\rvert 
       \leq  \sum_{m=0}^r \binom{r}{m} \left|a^{(m)} (\xi)\right|  \left|\phi^{(r-m)}(-\xi)\right| \ ,
    \end{align}
      where we use a similar procedure from Eq.~\eqref{eq: k^n combinatorics factors} to deal with the combinatorial degeneracy. 
      We then split the sum into $m=0$ case and $m \neq 0$ cases to evaluate the upper bound, i.e. 
      it becomes
      \begin{align}
           | g^{(r)}(\xi) | &\leq  |a (\xi)||\phi^{(r)}(-\xi)| + \sum_{m=1}^r \binom{r}{m} |a^ {(m)} (\xi)|  |\phi^{(r-m)}(-\xi)| \\
           &\leq \xi \epsilon_{k} \| \h \|^{r} +\sum_{m=1}^r \binom{r}{m}   \epsilon_{k} \| \h \|^{r-m} \ .
      \end{align}
      For the first term, we use the bound $ |a (\xi)| \leq \xi$ (Eq.~\eqref{eq: 1-exp bound}) and Eq.~\eqref{eq: phi tighter bound}.
      Similarly, for the second term, we use the bound $|a^{(m)} (\xi)| \leq 1$ and Eq.~\eqref{eq: phi tighter bound}.
        Furthermore, let us assume that $1 \leq \| \h \| $ which leads to $\lVert \h^r \rVert\geq \lVert \h^{r-1}\rVert$.
        Thus, it simplifies to
          \begin{align}
       | g^{(r)}(\xi) |
       \leq \epsilon_{k} \| \h \|^{r-1} +\sum_{m=1}^r \binom{r}{m}   \epsilon_{k} \| \h \|^{r-1} 
       &= \epsilon_{k} \| \h \|^{r-1}  \left(1+\sum_{m=1}^r \binom{r}{m} \right) \\
       &= 2^r \epsilon_{k} \| \h \|^{r-1}  \ . \label{eq: g^r bound}    \end{align}
      where we assume that $\xi\|\h\|\leq 1$ in the first inequality and we employ binomial identity in the last inequality.
    
      \end{enumerate}
Finally, observe that Eq.~\eqref{eq: k^4} can be upper-bounded by
\begin{align}
    | p^{(4)}(\xi)| &\leq 2 |g(\xi)| \times | g^{(4)}(\xi)| + 8 | g^{(1)}(\xi)|\times| g^{(3)}(\xi)|  + 6 | g^{(2)}(\xi)|^2 \ ,
\end{align}
where we use the fact that $ | g^{(r)} (\xi)| =  | g^{*(r)} (\xi)|$.
Using Eq.~\eqref{eq: g(t) <2} and Eq.~\eqref{eq: g^r bound}, it becomes
\begin{align}
    | p^{(4)}(\xi)|
    &\leq 64 \epsilon_{k} \|\h\|^{3}+   128 \epsilon_{k}^2 \|\h\|^{2}+96 \epsilon_{k}^2 \|\h\|^{2} \\
    &\leq 64 \epsilon_{k} \|\h\|^{3}+   128 \epsilon_{k} \|\h\|^{3}+96 \epsilon_{k} \|\h\|^{3} = 288 \epsilon_{k} \|\h\|^{3} \ ,
\end{align}
where we use the bounds $\epsilon_k^2 \leq \epsilon_k$ and  $|\h\|^{2}  \leq |\h\|^{3} $ in the last line.
Finally, the lower bound for $  p^{(4)}(\xi)$ is given by
\begin{align}\label{lower bound for k^4}
    | p^{(4)}(\xi)|\leq 288  \epsilon_{k} \|\h\|^{3} \implies p^{(4)}(\xi) \geq - 288  \epsilon_{k} \|\h\|^{3} \ .
\end{align}
    \end{enumerate}
Combining Eq.~\eqref{eq: simpler taylor series k}, Eq.~\eqref{eq: k^2 term} and Eq.~\eqref{lower bound for k^4} yields
\begin{align}
    p(t) & \geq 2t^2 E_k - 12 t^4 \epsilon_k \|\h\|^{3} \\
    &\geq 2t^2 \Delta \epsilon_k - 12 t^4 \epsilon_k \|\h\|^{3} \ ,
\end{align}
where we use the bound $E_k \geq \lambda_1 \epsilon_k=\Delta \epsilon_k$ as shown in Lemma.~\ref{lemma: Ek > lambda epsilon}.
As stated in Theorem.~\ref{th: fidelity convergence}, we set $t^2 = s = \dfrac{\Delta}{12 \|\h\|^3}$ as a constant for all $k$-th DB-QITE recursions.
Thus, $p(t)$ can be lower-bounded by
\begin{align}
    p(t) \geq 2\epsilon_k \left(\dfrac{\Delta^2}{12 \|\h\|^3}-\dfrac{6\Delta^2}{144 \|\h\|^3}\right) = \dfrac{\Delta^2 \epsilon_k}{12 \|\h\|^3} \ .
\end{align}
Substituting it into Eq.~\eqref{eq: epsilon and p(t) relation} yields
\begin{align}
     \epsilon_{k+1} \leq \epsilon_k -  (1-\epsilon_{k})\dfrac{\Delta^2 \epsilon_k}{12 \|\h\|^3} =\left(1-\dfrac{\Delta^2 F_k}{12 \|\h\|^3}\right) \epsilon_k \leq \left(1-\dfrac{\Delta^2 F_0}{12 \|\h\|^3}\right) \epsilon_k = q \epsilon_k \ ,
\end{align}
where $q =1- s  F_0 \Delta$ as mentioned in Theorem.~\ref{th: fidelity convergence}. 
Ultimately, the ground-state fidelity is given by
\begin{align}
\label{infidelity relation app}
      \epsilon_{k} \leq q^{k} \epsilon_0 & \implies  F_k \geq 1- q^{k}F_0 \ge 1-q^k
\end{align}
where we used that $\epsilon_0= 1- F_0 \leq 1$.

\end{proof}

\subsection{Runtime consideration}\label{runtime calculation}

We remark that Eq.~\eqref{eq: simpler taylor series k} implies the relation 
    $F_{k+1} = F_k + 2 E_k s + \mathcal O(s^2)$ stated as Eq.~\eqref{additive fidelity increase} above.
    Next, we provide a calculation of the depth in DB-QITE in the weaker setting of choosing $s$ via Eq.~\eqref{thm step duration}, so independent of encountered energy $E_k$ or energy variance $V_k$.

\begin{corollary} \label{cor:depth_scaling}
    For $L$ qubits, DB-QITE amplifies initial fidelity $F_0$ to desired fidelity $F_\text{th}$ in circuit depth
    \begin{equation}\label{eq: depth scaling app}
\mathcal O\left(L\left(\frac{1-F_0}{1-F_\text{th}}\right)^{2/(s  F_0\Delta)}\right)  \, .
\end{equation}
where  $\displaystyle s = \frac{\Delta}{12 \|\h\|^3}$ as defined in Theorem.~\ref{th: fidelity convergence}.
\end{corollary}
\begin{proof}
    We use Eq.~\eqref{infidelity relation app} to find $k$ such that
    \begin{align}
        \epsilon_{k} \leq q^{k} \epsilon_0 \le \epsilon_\text{th} := 1 -F_\text{th}\ .
    \end{align}
    After taking the logarithm on both sides we define
    \begin{align}
        k = \lceil \log( \epsilon_\text{th}/ \epsilon_0 ) / \log(q)\rceil\ .
    \end{align}
    Finally, we insert the lower bound to the query complexity estimate
    \begin{align}
        \mathcal O(3^k) 
        &= \mathcal O\left(( \epsilon_\text{th}/ \epsilon_0 )^{\log(3)/ \log(q)}\right)\\
        &= \mathcal O\left(( \epsilon_\text{th}/ \epsilon_0 )^{\log(3)/ \log(1-sF_0\Delta)}\right)\\
          &= \mathcal O\left(( \epsilon_0/\epsilon_\text{th}  )^{\log(3)/(sF_0\Delta)}\right)\ ,
    \end{align}
    where we use the identity $a^{\log_x(b)} = b ^{\log_x(a)}$ in the first line and $\log(1-x) \approx -x$ in the last line.
We bound the depth of each subroutine query as $\mathcal O(L)$ (i.e. dominated by reflections~\cite{zindorf2024efficient,zindorf_multicontrol_2025} as opposed to Hamiltonian simulations which can be done in $\mathcal O(1)$ time) so the overall depth is as stated.
\end{proof}

\section{Numerical simulations and gate-counts of DB-QITE using Qrisp}
\label{app Qrisp}

In this section, we provide numerical simulations and gate-counts for the compiled circuits for DB-QITE. A fully compilable implementation of DB-QITE requires careful management of quantum-algorithmic primitives. We utilize the Qrisp programming framework \cite{seidel2024qrisp}, whose automated memory management handles the exchange of auxiliary qubits across different modules without complicating the code.
Simulations are performed using Qrisp’s integrated \textit{statevector} simulator that leverages sparse matrices, enabling the simulation of DB-QITE for relatively large quantum systems. Apart from that, the simulator efficiently checks for potential factorization of state vectors, which can reduce the required resources significantly.

\subsection{Compilation for DB-QITE}

Our implementation utilizes the compilation primitives available within Qrisp to realize the DB-QITE unitary defined recursively in Eq.~\eqref{DB-QITE Uk}.
The DB-QITE unitary involves Hamiltonian evolutions and reflection gates, both of which need to be decomposed into the \texttt{CZ + U3} gate set for execution on available quantum hardware, and \texttt{Clifford + T} gate set for execution on fault-tolerant quantum hardware. Hamiltonian evolutions ($e^{i\sqrt{t}\h}$) are implemented using Trotter-Suzuki decomposition~\cite{PhysRevX.TrotterSuzukiError}. This is the standard approach to the task of Hamiltonian simulation and yields an efficient implementation when scaling the evolution duration. For the two-local Hamiltonians considered in our numerical studies, each term in the Trotter expansion can be efficiently simulated using a small number of CZ and single qubit rotations U3. To optimize the simulation depth, Qrisp's native implementation of Trotter-Suzuki decomposition employs a heuristic graph-coloring algorithm \cite{RLF} to parallelize the execution of Hamiltonian terms acting on non-overlapping qubits.

The reflection unitary is a multi-qubit controlled parameterized phase gate, assuming the quantum state is the tensor-product zero state; we see it by writing $e^{i\sqrt{t}\ket{0}\bra{0}}=\1+ (e^{i\sqrt{t}}-1)\ket{0}\bra{0}$ with matrix representation
\begin{equation}
    e^{i\sqrt{t}\ket{0}\bra{0}} = 
\begin{bmatrix}
e^{i\sqrt{t}} & 0 &  0 & 0 \\
0 & 1 & 0 &  0 \\
0 & 0 & \ddots &  0 \\
0  & 0 & 0 & 1 \\
\end{bmatrix}.
\end{equation}
This means that the first qubit is the target to which the operation $diag(e^{i\sqrt{t}},1)$ is applied, and the remaining $L-1$ qubits are the control qubits. 

\smallskip

Qrisp's native compilation routine for implementing the reflection gates ($ e^{i\sqrt{t} \ket0\bra0}$) is based on Ref.~\cite{balauca2022efficient}. The Qrisp implementation however deviates from the original implementation in a crucial way: The intermediate control values are no longer uncomputed via Margolus gates~\cite{song2003simplifiedtoffoligateimplementation} but instead Gidney's temporary logical AND \cite{Gidney2018halvingcostof}. This reduces the CZ count of the uncomputation from 3 to 0.5 (the CZ count only needs to be performed in half of the cases) and the T count from 4 to 0.

The above described technique is suitable to implement a multi-controlled X gate, the algorithm requires a multi-controlled phase gate. In principle, this can be achieved by wrapping two MCX gates around a phase gate and an ancilla. However, there is an even better way to achieve this. For this notice, that Balauca's MCX is a conjugation, i.e. of the form
\begin{equation}
    MCX = V W V^\dagger.
\end{equation}
The $V$ operator (un) computes the ancilla value, whilst the $W$ operator performs a simple Toffoli gate. We modify this structure in the following way towards implementing a multi-controlled phase gate:
\begin{equation}
    MCP(\Phi) = (V W) P(\Phi) (V W)^\dagger.
\end{equation}
In words: We use Balauca's ancillae computation method to reduce the control value into a single ancilla qubit, which is in the $\ket{1}$ state if the required conditional is given. On this ancilla, we execute the phase gate and conclude with the uncomputation of all involved ancillae.

A recently published, viable alternative is the compilation of multi-controlled X-gates using conditionally clean ancillae (CCA) \cite{khattar2024riseconditionallycleanancillae}. This method exhibits a similar conjugation structure, but in this case the $W$ operator requires a constant number of ancillae. The same technique as above can be used to implement a multi-controlled phase gate. The resulting procedure requires less ancillae, but it cannot use Gidney's logical AND for uncomputation. This fact drives the T-cost ($8n - 13$ for CCA and $4n -1$ for Balauca). As discussed in \cite{Gidney2018halvingcostof}, blocking an ancilla qubit for a given amount of space-time gives rise to an opportunity cost in the T count because the blocked ancillae cannot be used for T-state production. However, this cost is highly dependent on the space-time volume required to produce T states, which might change significantly until fault-tolerant quantum devices become available. The answer to the question which of these options proves advantageous therefore stays open.

For numerical simulations we utilize the latter procedure as it requires only a constant number of ancillae and handling a smaller amount of qubits is helpful for simulations.

\subsection{Numerical simulations of the Heisenberg model and initial states}

Let us consider DB-QITE applied to the antiferromagnetic Heisenberg model 
\begin{align}
\label{hamiltonian}
    \h = \sum\limits_{i=1}^{L-1}(X_iX_{i+1}+Y_iY_{i+1}+Z_iZ_{i+1})
\end{align}
where $L$ is the number of qubits. 
For the numerical simulations, we consider $L\in\{10,12,14,16,18,20\}$. Hamiltonian simulation, i.e. the unitary $e^{i\tau H}$, is implemented via the second-order Trotter-Suzuki formula with 2 steps. 
We remark that for some cases, one layer of the second-order Trotter-Suzuki formula could suffice.
In terms of runtime, roughly, each layer has linear runtime just like the reflection unitaries. 
In each step, we perform 2 Hamiltonian simulations and 1 reflection unitary; hence, taking 1 rather than 2 layers of the Trotter-Suzuki formula would lead to a shorter runtime, while in cases where $K$ layers are required, the runtimes presented below will scale in proportion to $K$.

In every DB-QITE step, we use a 20-point grid search to find the $s_k$ that yields the best energy gain.
Additionally, Eq.~\eqref{eqGC} has an approximate invariance 
\begin{align}
    G_{\alpha \beta s}(\hat A /{\alpha}, \hat B /{\beta}) = e^{i\sqrt{s \beta/\alpha}\hat A}e^{i\sqrt{s \alpha/\beta}\hat B}    
    e^{-i\sqrt{s\beta/\alpha}\hat A}
    e^{-i\sqrt{s\alpha/\beta}\hat B} =e^{-s[\hat A, \hat B]} +\mathcal O(s^{3/2})\ ,
\end{align}
but rescaling by $\alpha,\beta$ can influence the approximation constant in $\mathcal O(s^{3/2})$~\cite{double_bracket2024}. 
Specifically, we have
\begin{align}
    \|G_{\alpha \beta s}(\hat A /{\alpha}, \hat B /{\beta})-e^{-s[\hat A, \hat B]}\| \le s^{3/2}(\|[\hat A,[\hat A,\hat B]]\|/\sqrt{\beta/\alpha} + \|[\hat B,[\hat A,\hat B]]\|/\sqrt{\alpha/\beta}) \ .
\end{align}
We found empirically that setting $\alpha = 10$ and $\beta = 1$ allows us to identify time steps $s_k$ that yield better ground-state approximations than those obtained without reweighting the contributions of $\hat{H}$ and $\omega_k$.

We further remark that the DB-QITE steps involve two Hamiltonian simulations and one reflection. 
Depending on which of these components is more robust under noise, reweighting can be adjusted to improve noise resilience.
For instance, by shortening the Hamiltonian simulation time  and thus using fewer Trotter-Suzuki layers.
When optimizing for noise robustness or shorter circuits, the optimal reweighting strategy may differ from the one we used, which was designed to maximize energy decrease per step rather than circuit robustness.

Next, we discuss two types of initial states which \textit{a)} are designed analytically based on physical intuition so should remain meaningful for larger systems and \textit{b)} we consider a variational circuit which improves the initialization quality but may not be scalable.
Specifically, we consider \textit{a)} a tensor product of singlet states 
\begin{align}
    \ket{\text{Singlet}} = 2^{-L/4}(\ket{10}-\ket{01})^{\otimes L/2}
\end{align}
of consecutive qubits, and \textit{b)} a type of the so-called problem-specific Hamiltonian Variational Ansatz (HVA)~\cite{bosse_Heisenberg_2022, kattemolle_vanwezel_2022,wecker2015progress}
\begin{align}
\ket{\text{HVA}_p(\mathbf{t})} = \prod_{j=1}^{p}(e^{-it_{j,0} \h_0}e^{-it_{j,1} \h_1})\ket{\text{Singlet}}
\end{align}
where $\h_0, \h_1$ are obtained from $\h$ by restricting to summation over odd (even) indices $i$ such that all terms in $\h_0, \h_1$ commute. In numerics, we will use just one layer $p=1$ because it will be enough to illustrate the case of a significant improvement of the initialization.
In other words we use the ansatz
\begin{align}
\ket{\text{HVA}} = e^{-it_{0} \h_0}e^{-it_{1} \h_1}\ket{\text{Singlet}}= U_\text{HVA}\ket{\text{Singlet}}
\end{align}
whose meaning is to balance between the location of the singlets and putting the singlet sublattices into a superposition.

Fig.~\ref{fig:QITE Heisenberg energy and fidelity across L} shows the performance of DB-QITE for a range of system sizes $L=10,12,\ldots 20$ for both these initial conditions.
For all $L$, we find that just $k=2$ DB-QITE steps allow to reach $90\%$ fidelity when initializing with HVA.
We find across all considered system sizes that a better quality of initialization leads to a better DB-QITE performance, both in terms of ground-state fidelity and achieved energy.
However, for $L=20$ (which requires decreasing the Qrisp simulator's sparsity cut-off ratio to ensure high precision simulations) we find that the convergence requires more than $k=5$ steps to reach an ultra-high fidelity regime, i.e. to reach fidelity exceeding $F_k\ge99\%$ in that many-body system.
Having said that even for $L=20$, DB-QITE is able to prepare a state with energy within a fraction of the spectral gap $\Delta$ from the ground-state value $\lambda_0$.
This suggests that very short circuits could be enough to prepare states allowing to prepare very low-energy states for further physics studies on  a quantum computer.

Fig.~\ref{fig:QITE Heisenberg circuits gate counts} quantifies the runtime by counting gates of the DB-QITE circuits.
In near-term, the number of CZ gates matters most because usually operations an individual qubits are more coherent than gates acting on two qubits by coupling them.
On the other hand, for typical error-correction codes CZ gates can be implemented transversally which makes them cheaper than some of single-qubit rotations U3 which required $T$-gate resources. 
Typically the single-qubit count is about twice the CZ count.
For example, $k=2$ for $L=10$ leads to $F_2\approx 99\%$ and requires $N_\text{CZ}\approx 1.4 \times 10^3$ CZ gates and $N_\text{U3}\approx 2.3 \times 10^3$ single-qubit gates.
For $L=20$ we obtain $F_2\approx 92\%$ which requires $N_\text{CZ}\approx 3 \times 10^3$ CZ gates and $N_\text{U3}\approx 4.8 \times 10^3$ single-qubit gates.
Thus, the gate cost to achieve a high ground-state fidelity qualitatively doubled when doubling the system size, while it requires a much larger effort to prepare an ultra-high fidelity ground-state approximation for the larger system.

\newpage
\begin{figure}[h!]
    \centering
    \includegraphics[scale=0.65]{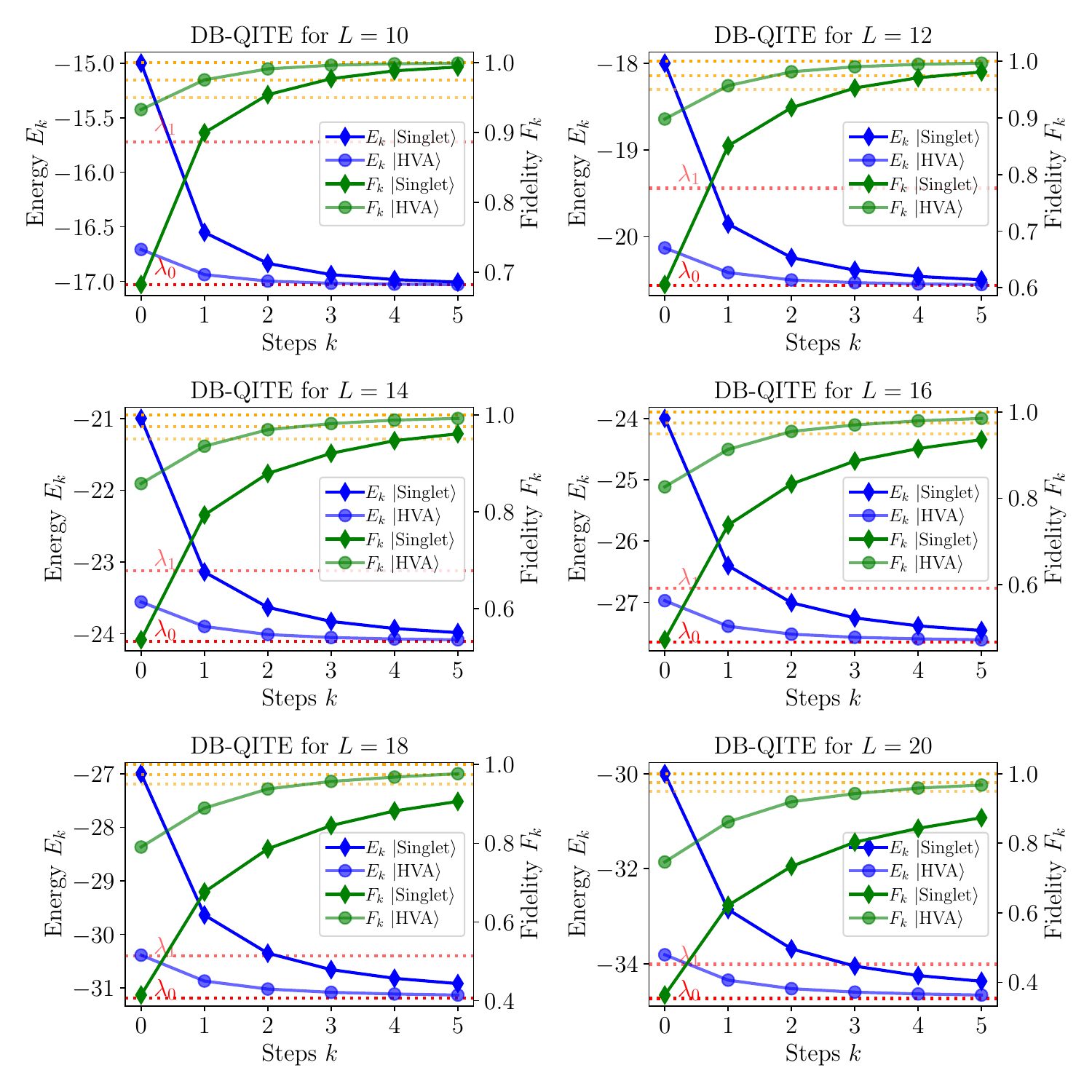}
    \caption{5 iterations of DB-QITE for $L$-qubit Heisenberg model. The figures show the energies $E_k=\langle \omega_k | \h | \omega_k \rangle$ and the fidelities $F_k=|\langle \omega_k | \lambda_0 \rangle|^2$ with the ground state $\ket{\lambda_0}$ for the $k$-th iterates $\ket{\omega_k}$ for initializations $\ket{\omega_0}=\ket{\text{Singlet}}$ and $\ket{\omega_0}=\ket{\text{HVA}}$. 
    The upper (dashed, orange) lines indicate the fidelities $1.0, 0.975, 0.95$, and the lower (dashed, red) lines indicate the ground state energies $\lambda_0$ and the first excited state energies $\lambda_1$.
    Let us make the following observations: 
    1) The $\ket{\text{HVA}}$ initializations achieve a higher fidelity to the ground state compared to the $\ket{\text{Singlet}}$ initializations, and the fidelity of the $k$-th iterates for  $\ket{\text{HVA}}$ consistently exceeds the respective fidelities for $\ket{\text{Singlet}}$. However, this gap narrows significantly for an increasing number of DB-QITE steps $k$. Therefore, the choice of the particular initial state can be less impactful as long as there is \textit{some} overlap with the ground state.
    2) The optimal evolution times $s_k$ found iteratively by a 20-point grid search tend to decrease with $k$, 
    while remaining orders of magnitude larger than guaranteed by Theorem \ref{th: fidelity convergence}. In particular, large initial evolution times $s_0$ yield a rapid decrease in energy in the first step of DB-QITE. Similarly, we observe that the reached fidelities $F_k$ are much larger than guaranteed by Theorem \ref{th: fidelity convergence} by optimizing the $s_k$ durations~\cite{xiaoyue2024strategies}. For example, for $L=10$ we empirically observe $F_k>1-q^k$ for $q=0.5$, while Theorem \ref{th: fidelity convergence} would merely provide a theoretical lower bound with $q$ being close to 1.
    3) The fidelity $F_k$ with the ground state $\ket{\lambda_0}$ reached for the $k$-th iterate decreases with increasing number of qubits $L$, e.g., for 5 steps of DB-QITE with $\ket{\text{HVA}}$ initialization, fidelities  $F_5=0.999$ for $L=10$, and $F_5=0.967$ for $L=20$ are achieved.
    }
    \label{fig:QITE Heisenberg energy and fidelity across L}
\end{figure}
\begin{figure}[h!]
    \centering
    \includegraphics[scale=0.55]{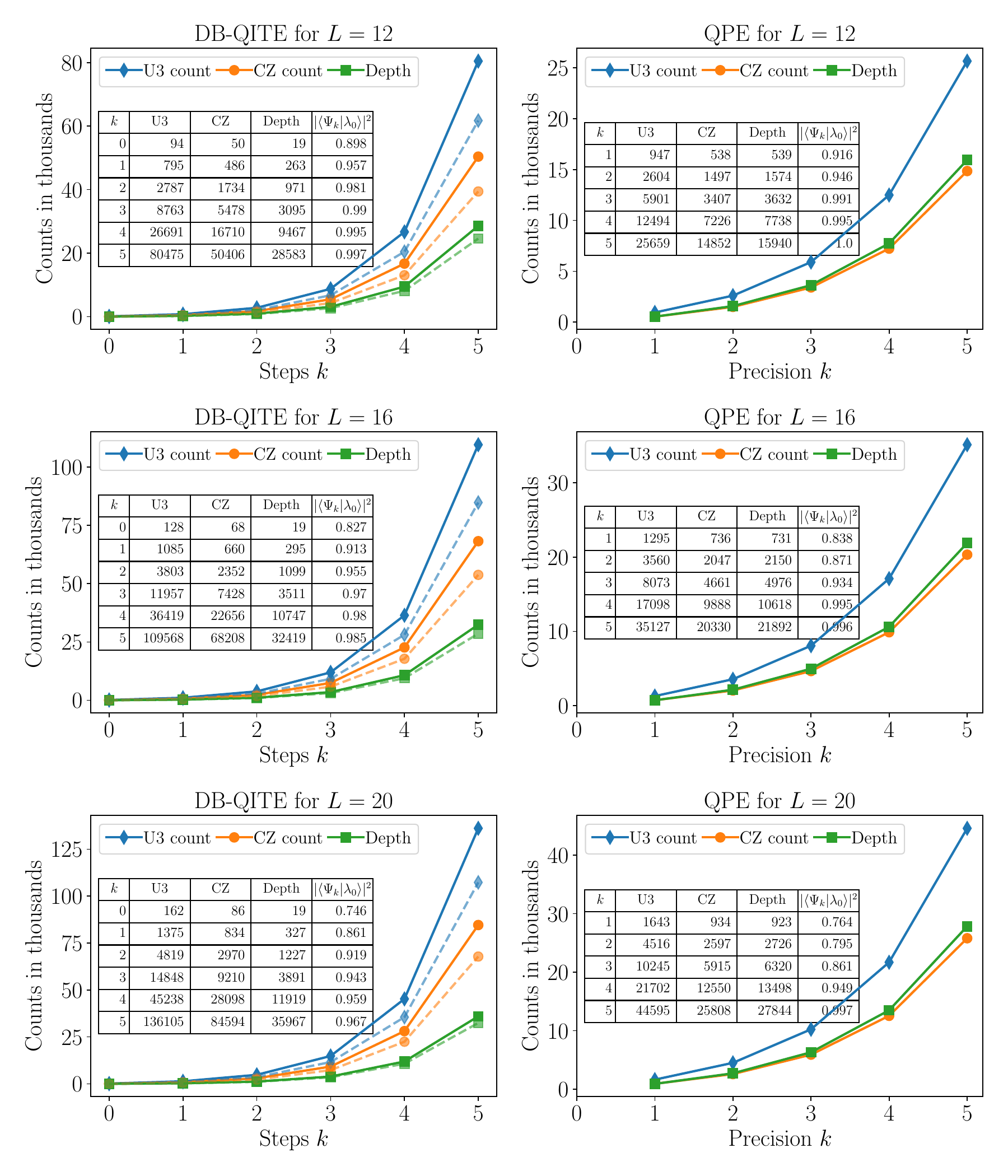}
    \caption{Gate counts, circuit depth and ground-state fidelity for DB-QITE (left column) and QPE (right column) for the Heisenberg model with $L=12,16$ and $20$ qubits for $\ket{\text{HVA}}$ and $\ket{\text{Singlet}}$ (dashed) as initializations. 
    For $L=12$ the ground state is prepared with fidelity $F_3 = 99\%$ with $k=3$ steps of DB-QITE. The circuit implementing this requires $N_\text{CZ}\approx 5.5 \times 10^3$ CZ gates and $N_\text{U3}\approx 8.8 \times 10^3$ single-qubit gates.
    For $L=20$ the ground state is prepared with fidelity $F_2 = 92\%$ with $k=2$ steps of DB-QITE. The circuit implementing this requires $N_\text{CZ}\approx 3 \times 10^3$ CZ gates and $N_\text{U3}\approx 4.8 \times 10^3$ single-qubit gates.
    The structure of DB-QITE's circuits is more homogeneous and e.g. for $k=4$ DB-QITE has circuit depth $D_4 \approx 1.2\times 10^4$ reaching $F_4 \approx 96\%$ which is lower circuit depth of QPE with $k=4$ precision qubits given by $D_4\approx 1.3 \times 10^4$ giving only $F_4 \approx 95\%$.
    When considering long enough circuits then QPE reaches enough precision ($k$ for QPE denotes the number of auxiliary qubits encoding precision) and under the assumption of all-to-all connectivity  outperforms DB-QITE by reaching higher fidelity through less gates and lower circuit depth.
    }
    \label{fig:QITE Heisenberg circuits gate counts}
\end{figure}

\newpage
\subsection{Conditions for quantum phase estimation (QPE) to outperform DB-QITE}\label{comparision QPE}
We next compare the runtime of DB-QITE to that of quantum phase estimation (QPE) for the ground-state preparation task.
Specifically, we will make a comparison to ground state preparation with phase estimation for known ground energy reviewed in Ref.~\cite{ge2019faster}.
We here aim to compare to the (likely) strongest algorithmic method in the regime when classical simulation is available.
In other words, for $L\le20$ qubits we can compute almost all relevant quantities to tightly parametrize QPE by having a $[0,1)$ spectrum.
In the case of $L\gg 20$ qubits such classical aid will not be present for a general Hamiltonian and we discuss the prospective runtime overheads later in this section and next introduce the precise method.

\paragraph{Description of the procedure for QPE.}
To set up the method we apply an affine transformation $\h\mapsto \h' = (\h-\lambda_0\1)/(\|\h\|-\lambda_0)$ to achieve that $\lambda_0=0$ and the spectrum of $\h$ is contained in $[0,1)$. 
We then perform controlled-evolutions using the rescaled Hamiltonian $U_j = \Lambda_j(e^{2\pi i 2^j \h'})$ where $\Lambda_j$ indicates controlling on qubit $j$.
The label $j$ encodes the digit of precision accessed.
The protocol then concludes by performing a measurement on all of the $k$ qubits encoding the precision and given the rescaling to $[0,1)$ the ground-state will correspond to measuring $0$'s.
Hamiltonian simulation, i.e. the unitary $e^{i t H}$, is implemented via the second-order Trotter formula with 2 steps for both methods.
The control $\Lambda_j$ is implemented by standard compilation methods which are natively available in Qrisp.
The success probability of preparing the (exact) ground state with QPE is upper bounded by the fidelity $F_0=| \langle \omega_0 | \lambda_0 \rangle |^2$ of the initial state with the ground state.

\paragraph{QPE vs. DB-QITE for singlet initialization.}
Fig.~\ref{fig:QITE QPE Heisenberg} summarizes the results for the $\ket{\text{Singlet}}$ initial condition.
First, for DB-QITE we find a gradual increase of recursion steps required to reach a threshold fidelity.
Second, the gate count is exponential in the number of steps.

For $L=10$, DB-QITE with $k=2$ steps provides ground-state fideltity $F_2\approx95\%$ requiring $N_\text{CZ}\approx 1.1\times 10^3$ CZ gates and $N_\text{sq}\approx 1.7\times 10^3$ single-qubit gates. QPE with $k=2$ precision qubits provides ground-state fideltity $F_2\approx 92\%$ with a success probability of $P\approx 76\%$ requiring $N_\text{CZ}\approx 1.2\times 10^3$ CZ gates and $N_\text{sq}\approx 2.1\times 10^3$ single-qubit gates lagging the results for DB-QITE. Importantly, $k=2$ requires about a thousand two-qubit gates which is within the reach of commercially available quantum hardware.

When aiming for higher fidelities, QPE starts to outperform DB-QITE:
For $L=10$, DB-QITE with $k=3$ steps provides ground-state fideltity $F_3\approx98\%$ requiring $N_\text{CZ}\approx 3.5\times 10^3$ CZ gates and $N_\text{sq}\approx 5.5\times 10^3$ single-qubit gates. QPE with $k=3$ precision qubits provides ground-state fideltity $F_3\approx 100\%$ with a success probability of $P\approx 70\%$ requiring $N_\text{CZ}\approx 2.7\times 10^3$ CZ gates and $N_\text{sq}\approx 4.7\times 10^3$ single-qubit gates.

With increasing number of qubits $L$, and decreasing initial fidelity $F_0$, the comparison further shifts in favor of QPE:
For $L=20$, QPE with $k=5$ precision qubits provides ground-state fideltity $F_5\approx 99\%$ with a success probability of $P\approx 37\%$ requiring $N_\text{CZ}\approx 2.6\times 10^5$ CZ gates and $N_\text{sq}\approx 4.4\times 10^5$ single-qubit gates. For DB-QITE the fideility $F_5$ remains well below $90\%$.
Unfortunately it is unrealistic to expect that this performance of QPE compared to DB-QITE would be sustained in cases when the ground-state energy is not known exactly.
In that case the spectrum would not be scaled as tightly and QPE would require more precision qubits and thus a much larger runtime.
See Ref.~\cite{ge2019faster} for analytical estimates of the overheads in the case where $\lambda_0$ is not known.

\paragraph{QPE vs. DB-QITE for inaccurate spectrum rescaling.}
In addition to results of QPE applied with the Hamiltonian $\h'$ rescaled to the $[0,1)$ interval, Fig.~\ref{fig:QITE QPE Heisenberg} presents results for $\h'/2$ rescaled to $[0,1/2)$ and $\h'/10$ rescaled to $[0,1/10)$.
This models a situation where QPE receives a strong guidance about the ground-state energy $\lambda_0$ but the norm $\|\h\|$ is over-estimated by a factor of $2$ and $10$, respectively.
Such over-estimation could arise, e.g. from using the triangle inequality for a decomposition of $\h$  into a sum of terms whose norm is known.
For example, for the Heisenberg model we have
\begin{align}
    \|\h\| \le \sum\limits_{i=1}^{L-1}\|X_iX_{i+1}+Y_iY_{i+1}+Z_iZ_{i+1}\| \le \sum\limits_{i=1}^{L-1}(\|X_iX_{i+1}\|+\|Y_iY_{i+1}\|+\|Z_iZ_{i+1}\|)\le 3(L-1)
\end{align}
which would give $\|\h\| \le 57$ for $L=20$ which is very loose but represents a difficulty that would appear for use-cases of quantum computation for difficult optimization instances in materials science.
We find that if one knows the exact values of $\lambda_0$ and $\|\h\|$ then QPE outperforms DB-QITE on platforms which have all-to-all interactions.
Thus, QPE requires stronger assumptions compared to DB-QITE and in near-term the latter will be advantageous.

We stress that the rescaling will be even worse when $\lambda_0$ is not known exactly.
Then the spectrum will not be rooted at $0$ and thus even more qubits will be required to resolve the energy with sufficient precision.
Moreover, the precise value of the ground-state energy will not be known and given that the initial ground-state fidelity may not be known then QPE might need to be executed for very many instances in order to allocate sufficient probability of observing measurement corresponding to the ground-state.

This should be compared with DB-QITE which is a quantum algorithm of a different type and thus requires no such estimates.
Our numerical results are computed based on evaluations of $E_k$ for 20 different values of $s_k$.
On quantum hardware this would require estimation of $E_k$ from data which is a well-studied problem and can be expected to be relatively robust (e.g. in comparison to evaluating gradients of $E_k$ with respect to microscopic parameters of the circuit as opposed to the macro-parameter $s_k$).
The overhead for the estimation $E_k$ will scale with $k$ and given that $k$ must be kept at minimum it should not affect the runtime substantially.

\paragraph{QPE vs. DB-QITE for pretrained initializations.}
These observations remain in place also for the variational initializations $\ket{\text{HVA}}$ presented in Fig.~\ref{fig:QITE QPE Heisenberg}.
Again, if exact rescaling is in place then QPE outperforms DB-QITE in reaching the ultra-high fidelity $F_5\ge 99\%$.
On the other hand, we see that DB-QITE benefits from an improved initialization more than QPE.
Specifically, for $L=20$ both DB-QITE with $k=3$ steps and QPE with $k=4$ precision qubits reach a fidelity of about $95\%$ through a similar runtime.
We remark that the steps are different in nature for these methods but they qualitatively coincide because for each there is a notion of exponential convergence at play.
To reach fidelity $F_2\ge 90\%$ DB-QITE is the preferred method because $k=2$ steps suffice while for $k=2$ QPE has $F_2 \ll 90\%$ and so DB-QITE has significantly faster runtime to reach the threshold of $90\%$ fidelity than QPE.

Finally, we comment that the circuit used to pretrain the initialization is a generalization of the $\ket{\text{Singlet}}$ initial state because it uses Hamiltonian evolutions that rotate towards singlets on alternating even and odd sublattices while $\ket{\text{Singlet}}$ implements them on only one sublattices.
However, despite the clear advantage of being informed by the physics of the model, the ansatz we use for training is limited in expressivity as it uses a low-depth circuit.
This showcases the strength of DB-QITE whose operation through quantum computation implements the intuition about the physics of the system in that DB-QITE implements steps along the Riemannian gradient which reduces the energy cost-function.

\paragraph{Discussion of susceptibility to noise.}
Finally, we comment in more detail about the robustness to experimental noise on near-term hardware.

Firstly, on architectures which do not have all-to-all interactions the runtime of both methods will be increase, generally DB-QITE should be less affected.
Specifically, the native connectivity must be used to route using SWAP gates connections needed for the required non-local gates~\cite{li2019tackling,yamamoto2023error,Efthymiou2024qibolabopensource}.
Ref.~\cite{zindorf_multicontrol_2025} shows that the reflection gates can be compiled with linear scaling and the proportionality constant increases only slightly even if the hardware architecture has only nearest-neighbour connectivity.
There is a difference, however, in that QPE needs routing from the qubits encoding the system to the qubits encoding the digits of the energy read-out.
This difference will show up especially in cases when the connectivity of the input Hamiltonian matches the hardware architecture.
In that case, Hamiltonian simulation in DB-QITE will have no routing overheads while the controlled Hamiltonian simulation of QPE will require it.
Additionally, with the constructions in Ref.~\cite{khattar2024riseconditionallycleanancillae}, the reflections can be executed with nearest-neighbor gates in linear time.

Besides overheads, the methods differ in their reliance on two-qubit gates.
For QPE the reliance is fundamental and the qubits responsible for precision will not faithfully represent the energy of the system if the controlled unitary operations are not close enough to perfection.
This could lead to measurements that point to wrong bits showing up in the energy representation.
In turn, for DB-QITE miscalibration of gates has a relatively transparent outcome: Two-qubits gates are used to implement the local terms of the Hamiltonian and if these gates have errors then we will be implementing DB-QITE with a different Hamiltonian $\h'$, albeit with the same locality; see Ref.~\cite{robbiati2024double} for a ground-state preparation DBQA which does not involve reflections but rather local diagonal Hamiltonians which are similarly robust.
The energy obtained with such deformed Hamiltonian will be transparently close to the energy of the ideal implementation up to an error given by $\varepsilon =2\|\h-\h'\|$.
In particular, for stable phases of matter in condensed-matter physics applications this robustness of DB-QITE could allow for useful near-term explorations.
\paragraph{Runtime gains for QPE when warm-starting with DB-QITE.}
Fig.~\ref{fig:QITE QPE together} showcases that using $k=2$ steps of DB-QITE on either $\ket{\text{Singlet}}$ or $\ket{\text{HVA}}$ systematically improves the performance of QPE.
Specifically for any number of precision qubits, the fidelity of QPE to the ground-state is higher when warm-starting using DB-QITE.
This is the most pronounced for $k=1$ or $k=2$ precision qubits where the QPE output fidelity is essentially the fidelity of the initialization.
For $k\ge3$ the QPE circuits begin increasing the fidelity but warm-starting gives systematically better outputs.
The gate counts shown in Fig.~\ref{fig:QITE QPE together} highlight the strength of DB-QITE in that $k=1$ or $k=2$ steps of DB-QITE have very short runtime and so the overhead of warm-starting is very small compared to the total runtime of QPE for $k\ge 3$.

When considering warm-starting, the case of using $k=4$ precision qubits in QPE is key.
This case represents the situation where QPE is already improving the fidelity but is limited in precision from reaching the ultra-high fidelity that we find for $k=5$.
In use-cases of quantum computation such situation will be the main challenge in that circuit depth could limit QPE from running for large numbers of precision qubits because the QPE runtime grows exponentially with $k$.
We see that for $k=4$ an improved warm-start using DB-QITE improves QPE output to about $50\%$ between QPE with $k=4$ and $k=5$.
The depth cost of running QPE for $k=5$ increases very strongly compared to $k=4$, in contrast, the sizeable improvement using DB-QITE comes at almost no overhead in comparison.

Moreover, warm-starting with DB-QITE increases the fidelity of the initial state with the ground state, hence improves the success probability of preparing the ground state with QPE.
We conclude that warm-starting using DB-QITE has a significant potential for improving the performance of QPE in practice.

\begin{figure}[htbp]
    \centering
    \includegraphics[scale=0.5]{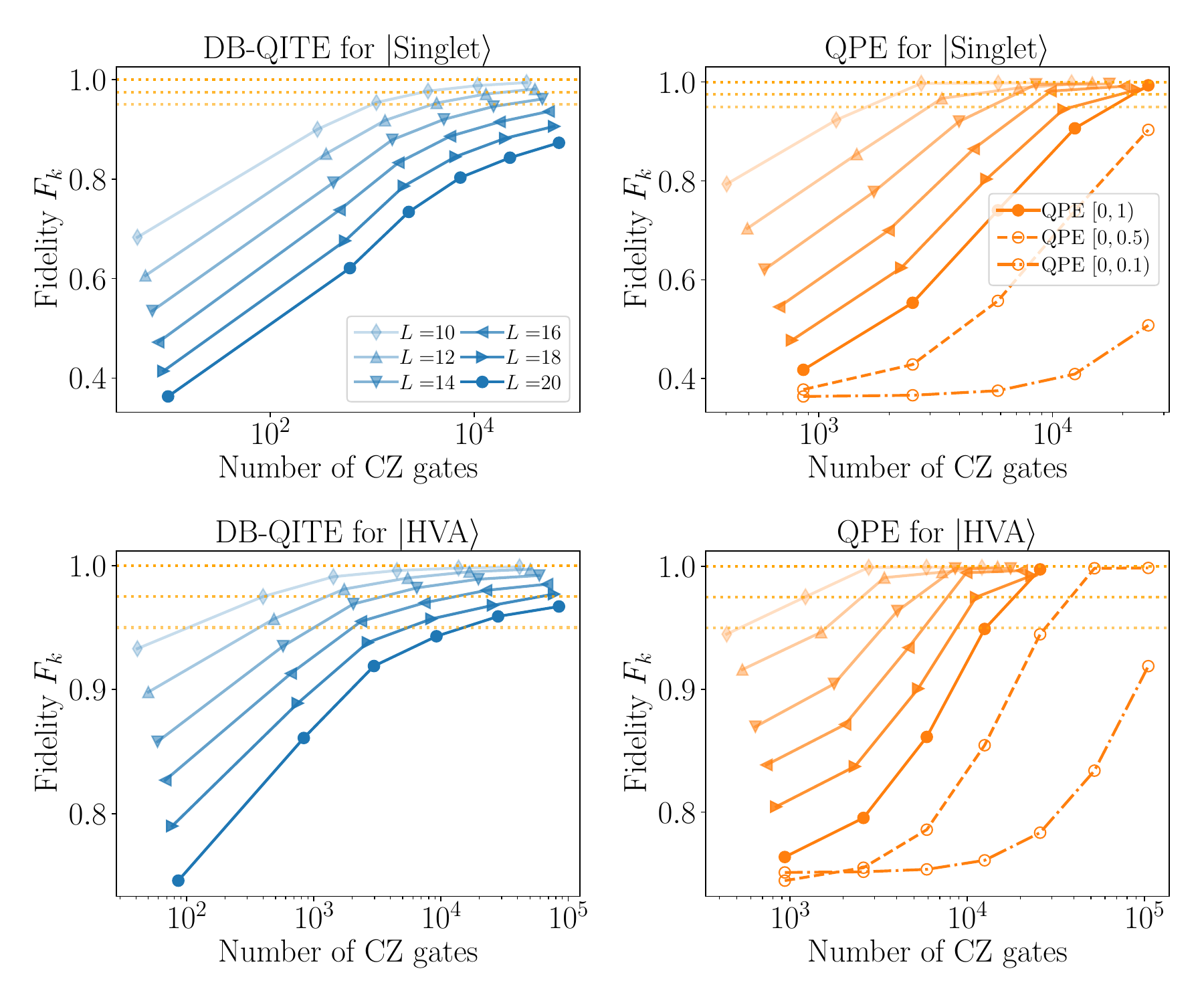}
    \caption{Comparison between DB-QITE and QPE for $L$-qubit Heisenberg model for $\ket{\text{Singlet}}$ (top) and $\ket{\text{HVA}}$ (bottom) initializations. 
    Fidelities $F_k$ with the ground state $\ket{\lambda_0}$ and number of CZ gate are shown as a function of steps $k=0,1,\dotsc,5$ for DB-QITE (left) and for QPE (right) as a function of the number of precision qubits $k=1,2,\dotsc,5$. 
    The dashed-orange lines indicate the fidelities $1.0, 0.975, 0.95$. For example, for 10 qubits just $k=4$ steps of DB-QITE are required to reach almost perfect fidelity, which is achieved with QPE  with $k=3$ precision qubits.
    Both methods require exponential gate counts but QPE converges faster than DB-QITE. For $\ket{\text{HVA}}$ initializations and $L=20$ qubits, a fidelity $F\approx 95\%$ is reached with $k=3$ steps of DB-QITE, or QPE with $k=4$ precision qubits, requiring $N_{\text{CZ}}\approx 10^4$ CZ gates in both cases. While DB-QITE deterministically prepares the (approximate) ground state, QPE has a success probability of $P\approx 79\%$. If QPE suffers from an inaccurate spectrum rescaling to $[0,0.5)$, the fidelity for $k=4$ precision qubits drops to $F\approx 86\%$ and an additional precision qubit is required. In this scenario DB-QITE clearly outperforms QPE.
    }
    \label{fig:QITE QPE Heisenberg}
\end{figure}

\begin{figure}[htbp]
    \centering
    \includegraphics[scale=0.5]{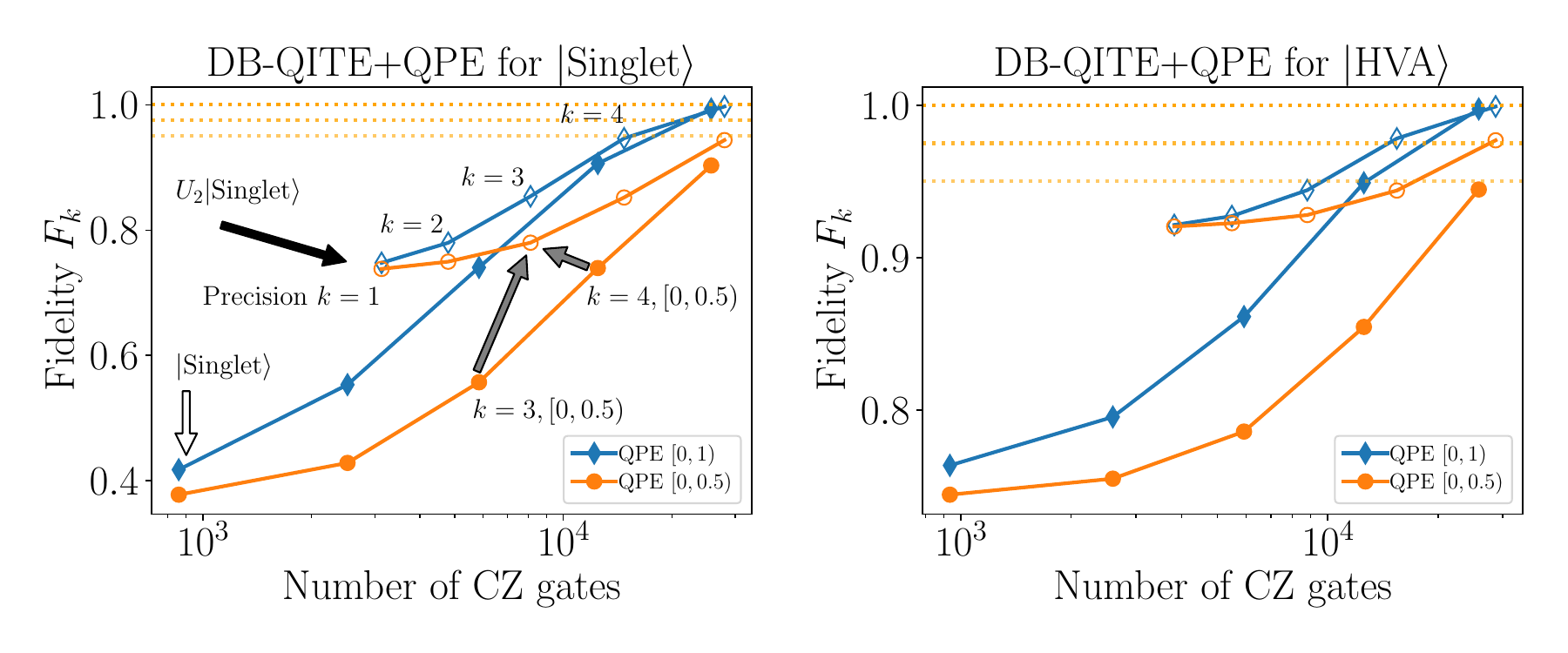}
    \caption{Combining DB-QITE and QPE: A comparison for QPE for $20$-qubit Heisenberg model with $\ket{\text{Singlet}}$, $\ket{\text{Singlet+QITE}}$ (left), $\ket{\text{HVA}}$, $\ket{\text{HVA+QITE}}$ initial states (right). Warm-starting with DB-QITE is achieved by performing $k=2$ steps of DB-QITE from the $\ket{\text{Singlet}}$ and $\ket{\text{HVA}}$ initial states, respectively.
    The figures show fidelities with the ground state $\ket{\lambda_0}$ and numbers of CZ gates as a function of the number of precision qubits $k=1,2,\dotsc,5$. 
The Hamiltonian $\h$ is rescaled by an affine transformation such that the spectrum in contained is $[0,1)$ (blue diamonds). 
For a more realistic scenario where the norm of the operator is not known exactly, the Hamiltonian is rescaled such that its spectrum is contained in $[0,0.5)$ (orange circles). In this case, additional precision qubits for the QPE are required and in effect the QPE runtime in order to reach $F_k \ge 90\%$ doubles. 
Unless QPE has reached full precision then DB-QITE warm-start increases QPE's output fidelity.
Additionally, even in the case of QPE converging onto the solution (end points of blue curves with diamonds) the warm-start should be used because it increases the success probability which is upper bounded by the initialization fidelity of QPE.
}
    \label{fig:QITE QPE together}
\end{figure}

\end{document}